\documentclass[accepted]{uai2025} 
\usepackage[american]{babel}

\usepackage{natbib} 
\usepackage{mathtools} 

\usepackage{microtype}
\usepackage{graphicx}
\usepackage{subfigure}
\usepackage{booktabs} 
\usepackage{comment}

\newcommand{\KL}{\textrm{KL}}
\newcommand{\E}[1]{\underset{#1}{\mathbb{E}}}
\newcommand{\Var}[1]{\underset{#1}{\textrm{Var}}}
\newcommand{\weightedmean}{\widetilde{M}}
\newcommand{\empmean}{\widehat{M}}
\newcommand{\emptarget}{\widehat{M}}
\newcommand{\truemean}{M}

\usepackage{amsmath}
\usepackage{amssymb}
\usepackage{mathtools}
\usepackage{amsthm}

\usepackage[capitalize,noabbrev]{cleveref}

\theoremstyle{plain}
\newtheorem{theorem}{Theorem}[section]

\newtheorem{lemma}[theorem]{Lemma}

\theoremstyle{definition}

\newtheorem{assumption}[theorem]{Assumption}
\Crefname{assumption}{Assumption}{Assumptions}
\newtheorem{task}[theorem]{Task}
\theoremstyle{remark}

\title{How Balanced Should Causal Covariates Be?}

\author[{ }]{\href{mailto:skkaul@cs.cmu.edu}{Shiva Kaul}{}}
\author[1]{\href{mailto:manjmin6@gmail.com}{Min-Gyu Kim}{}}
\affil[1]{%
    Department of Biomedical Informatics\\ 
    Ajou University School of Medicine\\
    Suwon, Republic of Korea
}

\begin{document}

\maketitle

\begin{abstract}
Covariate balancing is a popular technique for controlling confounding in observational studies. It finds weights for the treatment group which are close to uniform, but make the group's covariate means (approximately) equal to those of the entire sample. A crucial question is: how approximate should the balancing be, in order to minimize the error of the final estimate? Current guidance is derived from heuristic or asymptotic analyses, which are uninformative when the size of the sample is small compared to the number of covariates. This paper presents the first rigorous, nonasymptotic analysis of covariate balancing; specifically, we use PAC-Bayesian techniques to derive valid, finite-sample confidence intervals for the treatment effect. More generally, we prove these guarantees for a flexible form of covariate balancing where the regularization parameters weighting the tradeoff between bias (imbalance) and variance (divergence from uniform) are optimized, not fixed. This gives rise to a new balancing algorithm which empirically delivers superior adaptivity. Our overall contribution is to make covariate balancing a more reliable method for causal inference.
\end{abstract}

\section{Introduction}
\label{sec:intro}
In an observational study, some unknown mechanism splits the participants between the treatment and control groups. This can lead to systematic differences between the two groups, biasing estimates of the treatment effect. A key assumption in  causal inference is unconfoundedness \citep{rosenbaum1983central}, which posits that biases of the selection mechanism can be wholly accounted for by the covariates $X_i \in \mathbb{R}^d$ observed for each participant $i$. This assumption is more plausible when the number of covariates $d$ is large. Under this assumption, any meaningful differences between the groups must manifest as discrepancies between their covariate statistics; eliminating such discrepancies therefore controls confounding. This motivates the the technique of covariate balancing, also referred to as balancing weights, calibration weights or minimal weights (surveyed by \citet{ben2021balancing, chattopadhyay2020balancing}). This technique is embodied by the following optimization program for weighting the treatment group; a similar program can be written for the control group.
\begin{align}
\label{abstractpenaltybalancing}
\min_{W}   & \quad \max_{1 \leq j \leq d} \Big| \sum_{\textrm{treated}\ i} W_i X_{ij} - \emptarget_j \Big|\ + \\
            & \quad \lambda \cdot \textrm{Divergence}(W\ ||\ \text{Uniform})\ \nonumber\\
\text{s.t.} & \quad W \textrm{ is a distribution on the treatment group } \nonumber
\end{align}
Here, $\lambda > 0$ is a regularization parameter and $\textrm{Divergence}$ is an $f$-divergence between probability distributions. $\emptarget \in \mathbb{R}^d$ is the target mean vector, typically taken to be the empirical mean $\sum_{i=1}^n \frac{1}{n} X_{ij}$  of the entire sample, both treatment and control. An equivalent, more common formulation constrains imbalance by a tolerance parameter $\tau \geq 0$: 
\begin{align}
\label{abstractbalancing}
\min_{W}   & \quad \textrm{Divergence}(W\ ||\ \text{Uniform}) \\
\text{s.t.} & \quad W \textrm{ is a distribution on the treatment group } \nonumber\\
            & \quad \Big| \sum_{\textrm{treated}\ i} W_i X_{ij} - \emptarget_j \Big| \leq \tau \quad \forall\ 1 \leq j \leq d \nonumber
\end{align}
This program expresses an intuitive tradeoff between covariate imbalance and overfitting. The constraint ensures that each (weighted) covariate mean in the treatment group matches the target $\emptarget$, up to tolerance $\tau$; in other words, the weighted treatment group looks roughly like an unbiased random sample, insofar as covariate means are concerned. However, tightly balancing the means ($\tau \approx 0$) may cause some weights to become very small, reducing the effective sample size. Furthermore, $\tau \approx 0$ is generally not feasible when $n < d$. Using the KL divergence leads to entropy balancing \citep{hainmueller2012entropy}, and using the $\chi^2$ divergence leads to stable balancing weights \citep{zubizarreta2015stable}.

The success of covariate balancing depends on choosing $\tau$ (or $\lambda$) well. Because covariate balancing isn't predictive --- it estimates an unknown population quantity --- techniques for tuning $\tau$ on held-out data (e.g.~cross-validation) aren't readily adapted \citep{chattopadhyay2020balancing}. Ideally, $\tau$ should be chosen to minimize the error of the estimate derived from the weights (i.e.~the radius of its confidence interval). Unfortunately, no valid, nontrivial finite-sample confidence intervals have been derived for covariate balancing. Instead, practitioners rely on heuristics and asymptotic analyses. In early work studying approximate covariate balancing, \citet{zubizarreta2015stable} noted that $\tau = 0.1$ had been adopted in observational study matching with standardized covariates, but that this setting was problem-dependent. As noted by \citet{ben2021balancing}, various asymptotic analyses \citep{hirshberg2019minimax, hirshberg2021augmented} support $\tau = \Theta(1/n)$. The most complete, concrete guidance is currently given by the selection algorithm of \citet{wang2020minimal}, which was simplified slightly by \citet{chattopadhyay2020balancing}. Unfortunately, this algorithm still offers no provable guarantees. Furthermore, it (impractically) requires resolving (\ref{abstractbalancing}), and performing $d$ distinct bootstrap procedures, for every candidate value of $\tau$. 

Choosing $\tau$ remains a serious practical hurdle for covariate balancing. For example, in the recent benchmark comparison by \citet{cousineau2023estimating}, entropy balancing and stable balancing weights could not be successfully applied to over 40\% of the studies, due to infeasibility issues involving $\tau$. These problems arose even with a moderate number of covariates: each study in the benchmark had $n = 4802$ participants and $d = 58$ covariates. In production-scale, observational healthcare studies which inform clinical practice \citep{suchard2019comprehensive, chen2021comparative}, it is common to have $n \ll d$, with $d$ in the tens or hundreds of thousands \citep{tian2018evaluating, zhang2022adjusting}. Such studies instead use propensity score methods (such as propensity score matching or inverse propensity weighting, per \citet{hripcsak2021drawing}), which may not achieve balance.

\subsection{Our Contributions}

This paper provides the missing finite-sample analysis of covariate balancing. We derive the first valid, nontrivial, finite-sample confidence intervals for the weighted estimates produced by covariate balancing. Our results hold when the observational data come from a different distribution than the target of our inference --- a scenario called covariate shift or target transportation \citep{dahabreh2020extending, josey2021transporting}. Furthermore, we analyze an extension of (\ref{abstractpenaltybalancing}) where $\lambda$ is optimized along with $W$. This leads not just to satisfying theoretical guarantees, but to a new balancing algorithm which demonstrates empirical benefits when hyperparameter settings are uncertain.

The starting point of our analysis is the simple observation that the program (\ref{abstractpenaltybalancing}) conceptually aligns with the PAC-Bayesian paradigm of generalization (initiated by \citet{mcallester1998some} and \citet{shawe1997pac}; recently surveyed by \citet{alquier2024user}). This approach was developed to upper bound the expected population risk of a ``posterior'' distribution over hypotheses which is computed from a ``prior'' distribution and the data. The population risk bound is lower for posteriors (in our setting, $W$) which maintain low divergence from the the prior (for us, the uniform distribution) but also achieve low empirical risk (for us, empirical imbalance). There are conceptual and technical incompatibilities between covariance balancing and existing PAC-Bayesian inequalities, so we derive a fresh analysis starting from the underlying variational representation of $f$-divergences \citep{donsker1983asymptotic, ruderman2012tighter, ohnishi2021novel}. 

\Cref{sec:analysis} presents a generalization of covariate balancing which involves dynamically-optimized regularization parameters. We prove a nonasymptotic confidence interval for the estimates produced by this program, which (as a special case) apply to standard covariate balancing. Despite its generality, this proof is short and simple, relying on basic convex analysis. We describe how this improves upon a naive confidence interval presented earlier in \Cref{sec:thechallenge}. \Cref{sec:experiments} presents a pair of carefully-designed simulations which demonstrate that our new algorithm breaks through an adaptivity barrier, using the data to correctly, automatically optimize regularization parameters. 
Finally, the experiments in \Cref{sec:realdataexperiments} show these adaptivity benefits carry over to (completely) real data.  

\section{Previous Work}

\subsection{Main Balancing Algorithms}
\label{sec:balancingalgs}

The following algorithms involve solving either (\ref{abstractpenaltybalancing}) or (\ref{abstractbalancing}), so they are the main baselines for this paper. 

\textbf{Entropy balancing} \citep{hainmueller2012entropy}: this solves (\ref{abstractbalancing}) with the KL divergence. $\tau$ is set to a constant with a default value of $1.0$. The convex dual problem is unconstrained, so the Levenberg–Marquardt algorithm is used to solve it. Since this algorithm depends on a reasonably-conditioned $d \times d$ Hessian, it doesn't work when $n < d$. Though easily circumvented by using a standard gradient algorithm on the primal, this issue portends some of the practical inadequacies of covariate balancing in high dimensions.

\textbf{Stable balancing weights} \citep{zubizarreta2015stable}: this solves (\ref{abstractbalancing}) with the $\chi^2$ divergence. By default, $\tau$ is chosen according to the algorithm of \citet{wang2020minimal} from a grid of 8 possible values between $0.0001$ and $0.1$.  

\textbf{Covariate balancing propensity score} \citep{imai2014covariate} is typically described as a kind of  logistic regression, with exact ($\tau = 0$) balance constraints imposed on the fitted inverse-propensity weights. When used to estimate $\mu_1$, the inverse-propensity weights are normalized if a constant covariate is included \citep{sloczynski2023covariate}. Under this normalization, it can be formulated as (\ref{abstractbalancing}) with a shifted KL divergence \citep{josey2021framework}. \Cref{appx:cbps} elucidates this connection and presents the relevant $f$-divergence.

\subsection{Other Balancing Algorithms}

\textbf{Approximate residual balancing} \citep{athey2018approximate}: this approach solves a mild variant of (\ref{abstractpenaltybalancing}). By default, $\lambda = 1$. However, the weights are not directly used to form the weighted estimate (\ref{weightedestimate}). Instead, an elastic net regression is fit to the observed outcomes, and the weighted sum of its residuals are added back into the treatment effect estimate. 

\textbf{Kernel methods} replace the maximum imbalance, over the $d$ covariates, with a maximum over smooth functions $f$ within a reproducing kernel Hilbert space. \citet{wong2018kernel} add a term which penalizes the nonsmoothness of $f$, which reduces the maximum over $f$ to finding the top eigenvalue of a ($W$-dependent) matrix; furthermore, this eigenvalue can be differentiated with respect to $W$. Whereas \citet{wong2018kernel} solve a variant of (\ref{abstractpenaltybalancing}), \citet{hazlett2020kernel,hazlett2014thesis} focuses on (\ref{abstractbalancing}). Exact balance ($\tau = 0$) is imposed on the subspace corresponding to the top $r$ eigenvalues of the kernel matrix, where $r$ is chosen to minimize a proxy for estimator bias. Some of the ideas in this paper can be generalized to the kernel setting. 

An alternative to minimizing a combination of imbalance and divergence, as in (\ref{abstractpenaltybalancing}) or (\ref{abstractbalancing}), is to minimize a covariate-aware distance between distributions. For example, \citet{kong2023covariate} minimize an integral probability metric, such as the Wasserstein distance. \citet{huling2024energy} minimize the energy distance. These distances are very stringent, so minimizing them ensures consistent estimates without parametric assumptions. Unfortunately, these distances are even more difficult to constrain when $n \ll d$. For example, the Wasserstein distance converges at a slow rate of $O(n^{-1/d})$. 

More algorithms are discussed in \Cref{appx:relatedwork}.

\subsection{Balancing Analyses}

Most mathematical analyses of covariate balancing, and of causal inference more generally, have focused on asymptotics. The most well-studied asymptotic desideratum is \emph{$\sqrt{n}$-consistency}, which means the error of the estimate is $O(1/\sqrt{n})$. When $d = o(n)$, the weights from (\ref{abstractbalancing}) are consistent in this sense \citep{wang2020minimal}. When, $d = \omega(n)$, some kind of sparsity assumption is necessary to achieve consistency. However, even under such assumptions, covariate balancing has not been generally proven $\sqrt{n}$-consistent. This motivated the work of \citet{athey2018approximate}, which augments covariate balancing with sparse regression. 

Other desirable asymptotic properties include \emph{double robustness}, which means the estimate is consistent if at least one of the following holds: the model for the outcomes is correct (e.g. linear in the covariates), or the model for the propensity scores is correct \citep{fan2016improving, zhao2017entropy}. \emph{Semiparametric efficiency} essentially refers to achieving the lowest possible variance among estimators making the same assumptions about either the outcomes or propensity scores \citep{chan2016globally, zhao2017entropy}. 

Despite the focus on asymptotics, some finite-sample guarantees have been proven for covariate balancing. To justify their algorithm for automatically selecting $\tau$, \citet{wang2020minimal} bound, up to some positive constants, the excess loss incurred by balancing too many covariates. \citet{bruns2022outcome} derive the optimal weights (i.e.~the optimal bias-variance tradeoff) in terms of mean squared error. However, these are phrased in terms of unknown population quantities. \citet{su2023estimated} prove a useful finite-sample decomposition of the inverse propensity weighting error.  

It should be noted that all theoretical results for covariate balancing, including this paper's finite-sample intervals,  depend upon causal assumptions which are strong and potentially unrealistic. (These will be discussed in \Cref{sec:assumptions}). Thus, a truly ``rigorous'' process of causal inference must perform additional sensitivity checks \citep{imai2010identification} and/or embed itself in meta-analytic inference which does not require such assumptions \citep{kaul2024meta}. In these circumstances, practitioners 
 often disregard the nominal coverage probability of confidence intervals, and just make use of point estimates \citep{schuemie2020confident}. 
We stress that the theoretical guarantees developed in this paper are presented not just for their own sake, but to help design new algorithms which are easier to use, demonstrate better adaptivity, and ultimately achieve lower error.

\subsection{Data-Driven Regularization}

The problem of optimizing regularization parameters in a data-dependent manner is a recurring topic in the PAC-Bayesian literature \citep{catoni2007pacbayes, alquier2024user}. As the latter survey notes, the subtlety of this issue has resulted in the publication of mistaken proofs. The fundamental difficulty arises from the fact that the thresholds in Markov's inequality and Hoeffding's lemma cannot be random. Recently, \citet{rodriguez2024more} derived a version of Catoni's classic PAC-Bayesian inequality which holds uniformly for all (potentially data-dependent) $\lambda$, with a small cost of $\log n / n$ in the confidence term. As \Cref{sec:analysis} demonstrates, we exploit specific aspects of covariate balancing to obtain results which may not be achievable in the most general PAC-Bayesian setting. 

\section{Preliminaries}

Suppose there are $n_0$ participants in the control group and $n_1$ participants in the treatment group. Each participant $i$ has data $(X_i, Y_i(0), Y_i(1), T_i)$. $X_i \in \mathbb{R}^d$ are covariates, possibly generated by computing features of underlying lower-dimensional covariates. $Y_i(0)$ and $Y_i(1)$ are the potential outcomes under no treatment and treatment, respectively. $T_i \in \{0, 1\}$ denotes whether the participant was actually treated. We observe only $(X_i, Y_i(T_i), T_i)$, leaving one of the potential outcomes unobserved. Our goal is to estimate the average treatment effect (ATE), which is the difference between the treatment-specific means $\mu_1$ and $\mu_0$: \begin{align}
\label{ate}
\textrm{ATE} = \underbrace{\mathbb{E}_\mathcal{T}\ Y(1)}_{\mu_1} - \underbrace{\mathbb{E}_\mathcal{T}\ Y(0)}_{\mu_0} 
\end{align}
The expectation in (\ref{ate}) may be over a different ``target'' distribution $\mathcal{T}$ than the one sourcing the data $\mathcal{S}$. If $\mathcal{S} \neq \mathcal{T}$ --- a phenomenon referred to as covariate shift or transportation --- we don't see individual data from $\mathcal{T}$, just a target vector $\emptarget \in \mathbb{R}^d$ formed by averaging $n$  samples from $\mathcal{T}$. (For example, imagine using observational data to emulate a randomized trial for which summary statistics are published, but individual patient data are unavailable). More commonly, $\mathcal{S} = \mathcal{T}$ and $\emptarget = \sum_{i=1}^n \frac{1}{n} X_i$ where $n = n_0 + n_1$. The following mean vectors --- weighted, target and true --- should, ideally, all be close to one another.
\begin{align}
\label{threemeans}
\weightedmean = \sum_{\textrm{treated } i} W_i X_i \hspace{1mm}
&&
\emptarget \in \mathbb{R}^d \hspace{1mm}
&&
\truemean = \mathbb{E}_\mathcal{T}\ \emptarget
\end{align}
As is typical in the covariate balancing literature, we focus on deriving an estimate $\hat{\mu}_1$ of $\mu_1$. Then, an estimate $\hat{\mu}_0$ of $\mu_0$ can be similarly produced, and the two can be combined into an estimate $\widehat{\textrm{ATE}}$ of ATE. 
\begin{task}[Weighted estimation of treatment-specific means]
\label{ourtask}
Let $\{(X_i, Y_i(T_i), T_i)\}_i$ be data from $\mathcal{S}$ and $\emptarget$ be an empirical target vector from $\mathcal{T}$. Use these to pick a probability distribution $W$ over the treated $\{i : T_i = 1\}$ and estimate: \begin{align} 
\label{weightedestimate}
\hat{\mu}_1 = \sum_{\textrm{treated}\ i} W_i Y_i(1)
\end{align}
Given $\alpha \in (0,\frac{1}{2})$, find $\epsilon$ such that $|\hat{\mu}_1 - \mu_1| \leq \epsilon$ with probability $1-\alpha$ over the randomness of the participant data. That is, $\hat{\mu}_1 \pm \epsilon$ is a $1-\alpha$ confidence interval for $\mu_1$.
\end{task}

\subsection{Mathematical Notation}
\label{sec:notation}

Given any $d$-dimensional real vector $v \in \mathbb{R}^d$, these are its $\ell_1$, $\ell_2$, and $\ell_\infty$ norms: \begin{align*}
||v||_1 = \sum_j |v_j| 
\hspace{5mm}
||v||^2_2 = \sum_j v_j^2
\hspace{5mm}
||v||_\infty = \max_j |v_j| 
\end{align*}
Let $f: [0, +\infty)\to(-\infty, +\infty]$ be a convex function such that $f(t)$ is finite for all $t > 0$, $f(1) = 0$, and $f(0)=\lim_{t\to 0^+} f(t)$. Let $U$ (sometimes expanded as  \textrm{Uniform}) be the uniform distribution over the treatment group. Let $W$ be any other distribution on the treatment group. Then the $f$-divergence between $W$ and $U$ is:
\begin{align}
\label{fdivergence}
\textrm{Divergence}(W\ ||\ U) &= 
\sum_{\textrm{treated i}} \frac{1}{n_1} f(n_1 \cdot W_i) 
\end{align}
Its (restricted) convex conjugate is defined for $Z \in \mathbb{R}^{n_1}$: \begin{align*}
\textrm{Divergence}^*(Z\ ||\ U)
&= \sup_{W} Z^\top W - \textrm{Divergence}(W\ ||\ U)
\end{align*}
where the supremum is taken over probability distributions $W$ \footnote{If the supremum were taken over the larger set of nonnegative vectors (i.e. measures) then $\textrm{Divergence}^*$ could be tidily expressed in terms of the conjugate $f^*$ of $f$. However, that leads to looser bounds \citep{ruderman2012tighter, ohnishi2021novel}.}. For example, the KL divergence between $W$ and $U$ is the difference of the Shannon entropies $H(U)$ and $H(W)$: 
\begin{align}
\label{klandentropy}
\KL(W\ ||\ U) &= \underbrace{\log(n_1)}_{H(U)} - \underbrace{\sum_{\textrm{treated}\ i} - W_i \log W_i}_{H(W)} \\
\KL^*(Z\ ||\ U) &= \log \frac{1}{n_1} \sum_{\textrm{treated}\ i} \exp Z_i \nonumber
\end{align}
The equality (\ref{klandentropy}) shows why entropy balancing, which maximizes $H(W)$, is a special case of (\ref{abstractbalancing}).

\subsection{Assumptions}
\label{sec:assumptions}

Much like \citet{athey2018approximate}, we assume that the outcomes are linear functions of the covariates. 
\begin{assumption}[Linear outcomes]
\label{asm:linearity}
Let $0 < k \leq d$ be some sparsity constant. There are $u_* \in \mathbb{R}^d$ and $v_* \in \mathbb{R}^d$ such that $||u_*||_1 \leq k$, $||v_*||_1 \leq k$. For all covariates $X \in \mathbb{R}^d$ and potential outcomes $(Y(0), Y(1))$ sampled from either $\mathcal{S}$ or $\mathcal{T}$, $Y(0) = u_*^\top X$ and $Y(1) = v_*^\top X$.
\end{assumption}
Note that we assume only ``approximate'' sparsity, with $\ell_1$ norm bounds on $u_*$ and $v_*$. By contrast, \citet{athey2018approximate} assumes an exact sparsity bound on the number of nonzero coordinates of $u_*$ and $v_*$. On the other hand, our assumption does not permit noise in the outcomes. Additive zero-mean noise can be handled by standard techniques, which we omit to highlight the novelties of our approach. Note that we permit $n \ll d$, where it is typically not difficult to fit the outcomes with near-zero error.

We also assume that the empirical target $\emptarget$ concentrates to the true $\truemean$ at a typical rate. If each of the $d$ covariates is bounded (i.e.~is normalized to take values in $[-1,1]$), and the samples are independent, then the following assumption is just a consequence of Hoeffding's inequality.
\begin{assumption}[Concentration of covariates]
\label{asm:concentration}
For constants $\alpha \in (0,\frac{1}{2})$ and $C > 0$, $||\emptarget - \truemean||_\infty \leq \frac{C \log d }{\sqrt{n}}$ with probability $1 - \alpha$.
\end{assumption}
Our nonasymptotic confidence interval depends on only the aforementioned assumptions. When extending it to an asymptotic consistency guarantee, we will need to invoke more traditional causal assumptions, in order to ensure the existence of (correct) inverse-propensity weights. As discussed in \Cref{appx:existenceofipw}, the following assumption is a consequence of strong ignorability \citep{rosenbaum1983central}. When $\mathcal{S} \neq \mathcal{T}$, a condition such as positive participation probability is also needed \citep{dahabreh2020extending}.
\begin{assumption}[Inverse-propensity weights]
\label{asm:optimalweights}
There is a probability distribution $W^*$ which is ratio-bounded (i.e. $\frac{\max_i W^*_i}{\min_i W^*_i} \leq R$ for some constant $R \geq 1$) and satisfies:\begin{align*}
\left|\left| \sum_{\textrm{treated}\ i} W^*_i X_{i} - \truemean \right|\right|_\infty = O\left(\frac{\log d}{n_1}\right)
\end{align*}
\end{assumption}
We will demonstrate consistency in the same high-dimensional regime as \citet{athey2018approximate}.
\begin{assumption}[Asymptotics]
\label{asm:asymptotics}
As $n \to \infty$,  $\frac{k \log d}{\sqrt{n_1}} \to 0$.
\end{assumption}

\subsection{A Naive Confidence Interval}
\label{sec:thechallenge}

To motivate our main results, consider the following naive attempt to construct a confidence interval. Expanding the definitions in (\ref{weightedestimate}) and (\ref{threemeans}), then applying H{\"o}lder's inequality in conjunction with  \Cref{asm:linearity,asm:concentration}, we can upper bound $\hat{\mu}_1 - \mu_1$. Then, we can repeat the same logic on $\mu_1 - \hat{\mu}_1$ to bound $|\hat{\mu}_1 - \mu_1|$.
\begin{align}
\hat{\mu}_1 - \mu_1
&= \sum_{\textrm{treated}\ i} W_i Y_i(1) - \mathbb{E}\ Y(1) \nonumber\\
&= \sum_{\textrm{treated}\ i} W_i v_*^\top X - \mathbb{E}\ v_*^\top X \nonumber\\
&= v_*^\top (\weightedmean - \truemean) \label{linofmeans}\\
&= v_*^\top(\weightedmean - \empmean) + v_*^\top(\empmean - \truemean) \nonumber\\
&\leq ||v_*||_1 ||\weightedmean - \empmean||_\infty + ||v_*||_1 ||\empmean - \truemean ||_\infty \nonumber\\
&\leq k \Big( ||\weightedmean - \empmean||_\infty + \frac{C \log d}{\sqrt n} \Big) \label{naiveinterval} 
\end{align}
Note that when $n < d$, there generally aren't weights which perfectly balance the covariates; the imbalance term will typically have to be (substantially) positive. Nonetheless, this interval blithely suggests weights which make $\weightedmean$ as close as possible to $\empmean$; there is no benefit to choosing $W$ with low divergence from uniform. Thus, the challenge for our analysis is: can we use the PAC-Bayesian machinery to (1) scale down the imbalance term, and (2) quantitatively explain the intuitive benefit of near-uniform weights?

\section{PAC-Bayesian Analysis}
\label{sec:analysis}
We propose the following program, which generalizes (\ref{abstractpenaltybalancing}). We call this \textbf{flexible covariate balancing} because the regularization parameter $\lambda$ is no longer fixed; it is optimized jointly with the weights $W$, along with a new parameter $\delta$. The appropriate tradeoff between imbalance and divergence is determined by the data.
\begin{align}
\label{jointbalance}
\min_{W,\lambda,\delta} \max_{z}\ & \big((1-\delta)k + \delta \beta_z\big) \max_{j} \Big| \sum_{\textrm{treated}\ i} W_i X_{ij} - \emptarget_j \Big|\ + \nonumber\\
            & (1-\delta) (C k \log d) / \sqrt{n}\ + \\
            & \lambda\ \textrm{Divergence}(W\ ||\ U)\ + \nonumber\\
            & \lambda\ \textrm{Divergence}^*\Big(\frac{\delta}{\lambda} [z Y_i(1) - \widehat{Y}_i(1)]_{\textrm{treated}\ i} \ \Big|\Big|\ U \Big) \nonumber\\ 
\text{s.t.}\ & W \textrm{ is a distribution on the treatment group } \nonumber \\
            & \lambda > 0 \hspace{6mm} 0 \leq \delta \leq 1 \hspace{6mm} z \in \{-1, 1\} \nonumber
\end{align}
Here, $\widehat{Y}_i(1)$ is a sort of fulcrum for $Y_i(1)$ which is calculated from $X_i$ and the largest-magnitude covariate of $\widehat{M}$: 
\begin{align*}
j_* = \textrm{argmax}_j |\emptarget_j|
\hspace{4mm}
s = \textrm{sign}\big(\emptarget_{j^*}\big)
\hspace{4mm}
\widehat{Y}_i(1) = - s \beta_z X_{ij^*}
\end{align*}
For each sign $z \in \{-1, 1\}$, $\beta_z \geq 0$ is the value of the following program, which is solved for each value $z \in \{-1,1\}$ in preparation for (\ref{jointbalance}). It can be rephrased as a linear program, as shown \Cref{appx:linearprogram}.
\begin{align}
\label{precursorprogram}
\max_{v_* \in \mathbb{R}^d, \truemean \in \mathbb{R}^d} & \quad \frac{-z}{||\emptarget||_\infty} v_*^\top M \\
\text{s.t.} & \quad ||v_*||_1 \leq k \hspace{6mm} ||\emptarget - M||_\infty \leq \frac{C \log d}{\sqrt{n}} \nonumber\\
            & \quad v_*^\top X_i = Y_i(1) \hspace{2mm} \textrm{for all treated}\ i \nonumber
\end{align}
The estimate derived from (\ref{jointbalance}) enjoys the following nonasymptotic guarantee on its error.
\begin{theorem}
\label{thm:confidenceinterval}
Fix $\alpha \in (0,\frac{1}{2})$. Let $\nu$ be the objective value of (\ref{jointbalance}) obtained by any feasible $(W, \lambda)$. Construct $\hat{\mu}_1$ from $W$ as in (\ref{weightedestimate}). Then, under \cref{asm:linearity,asm:concentration},  $|\hat{\mu}_1 - \mu_1| \leq \nu$ with probability $1-\alpha$.
\end{theorem}
In other words, the objective of (\ref{jointbalance}) is the radius of a confidence interval solving \Cref{ourtask}. $W$, $\lambda$ and $\delta$ are jointly optimized to shrink the confidence interval. This interval has guaranteed finite-sample coverage --- not just for the global solution of the program, but every feasible iterate. So, by keeping $\lambda$ and $\delta$ fixed, we obtain confidence intervals for plain covariate balancing (\ref{abstractpenaltybalancing}) as a special case.

\subsection{Intuition with KL Divergence}

At first glance, the program (\ref{jointbalance} is seemingly complicated and opaque, especially since it involves a precursor program (\ref{precursorprogram}). Let us gain some more concrete intuition for how it improves upon plain covariate balancing (\ref{abstractpenaltybalancing}), and how \Cref{thm:confidenceinterval} improves upon the naive interval (\ref{naiveinterval}). First note that, with $\delta = 0$ and $\lambda \to 0$, the naive interval is recovered. This pessimistic, oblivious bound may be useful asymptotically, if very non-uniform weights can, and must, be taken to achieve low imbalance. Larger $\delta$ enable a more optimistic, scale-sensitive analysis which explains the benefit of near-uniform weights in smaller samples.

At the other extreme $\delta = 1$, the imbalance term is scaled by $\beta_k$ rather than $k$. Typically $\beta_k \ll k$, so this scaling achieves the goal described in \Cref{sec:thechallenge}. To see when the $\beta_z$ are small, let us examine the program (\ref{precursorprogram}) which define them, focusing on the common case where $\emptarget \approx \truemean$. Without the equality constraints, the value of the program would be roughly $k$, attained by $v_*$ in the coordinate direction $-z \cdot e_{j^*}$. However, the equality constraints usually prevent this value from being attained. As the notation in this program suggests, it bounds the magnitude of $v_*^T M = \mu_1$, i.e. the average treatment outcome in the target distribution. We don't observe such outcomes, but by \Cref{asm:linearity}, we know they must be linearly consistent with the treatment outcomes we do observe. $\beta_z$ will be small when we can safely conclude that outcomes in the unseen target distribution aren't much larger than those in the data. (We have therefore controlled for the possibility that, under covariate shift, the outcomes are wildly different in the target distribution).

To improve upon the naive interval, the extra divergence-related terms cannot be too large. To gain intuition for these, consider (\ref{jointbalance}) in the specific case of the KL divergence (\ref{klandentropy}): 
\begin{align}
\label{jointklbalance}
\min_{W,\lambda,\delta} \max_{z}\  & \big((1-\delta)k + \delta \beta_z\big) \max_{j} \Big| \sum_{\textrm{treated}\ i} W_i X_{ij} - \emptarget_j \Big|\ + \nonumber\\
            & (1-\delta) (C k \log d) / \sqrt{n}\ + \\
            & \lambda\ \KL(W\ ||\ U)\ + \nonumber\\
            & \lambda \log \E{i \sim U} \exp \Big(\frac{\delta}{\lambda} (z Y_i(1) - \widehat{Y}_i(1)) \Big) \nonumber\\ 
\text{s.t.}\ & W \textrm{ is a distribution on the treatment group } \nonumber \\
            & \lambda > 0 \hspace{6mm} 0 \leq \delta \leq 1 \hspace{6mm} z \in \{-1, 1\} \nonumber
\end{align}
Recall that, as $\lambda \to 0$, the \textrm{logmeanexp} function approaches \textrm{max}. As $\lambda \to \infty$, it approaches \textrm{mean}. Without loss of generality, we can linearly center the treatment-group means of the outcomes $Y_i(1)$ and the covariate $X_{ij^*}$. So, in the benign case where $W \approx U$ achieves good imbalance, we can take $\lambda$ to be large, and keep both divergence terms near zero. In this way, the bound adapts to treatment groups which are empirically similar to the combined population.

In general, to keep imbalance low, we have to choose weights with nontrivial KL divergence. We can justify higher KL divergence if the data exhibit low variation around their mean. This makes intuitive sense: if the $Y_i(1)$ are all close to one another, then choosing extreme weights can't change $\sum_i W_i Y_i(1)$ as much. Whereas plain covariate balancing completely ignores the outcomes while choosing $W$, flexible covariable balancing makes use of these data to help optimize all the variables. This argues against the principle, advocated by \citet{rubin2008objective}, that outcome data should not affect the design of observational studies.

The bound depends on $k$, the $\ell_1$ norm of the linear function generating the outcomes. $k$ is usually unknown, but more standard regression techniques can be used to reason about it. For example, since the outcomes are assumed to be consistent with the linear function, a prospective value of $k$ is too small if it leads to nonzero training error.  

\subsection{Proof with KL Divergence}
\label{sec:proofwithkl}
The proof of \Cref{thm:confidenceinterval} is short and simple. The proof using the KL divergence contains all the ideas; extending it to all $f$-divergences, in \Cref{appx:otherfdivergences}, requires just a bit more notation. The proof begins with the Donsker-Varadhan formula for the KL divergence, which holds for all $W$: \begin{align}
\label{donskervaradhan}
\KL(W\ ||\ U) 
&= \sup_{\phi} \E{i \sim W}\ \phi(i) - \log \E{i \sim U} \exp \phi(i)
\end{align}
Here, the supremum is taken over all bounded functions $\phi$ from the set of treated $i$ to $\mathbb{R}$. Therefore, by restricting the supremum to any smaller set of functions, we get a lower bound on the KL divergence. We restrict to the following functions parameterized by $z \in \{-1, 1\}$, $\lambda > 0$ and $\hat{v} \in \mathbb{R}^d$. Recalling the notation for means in (\ref{threemeans}) and the definition of $v_*$ in \Cref{asm:linearity}:
\begin{align*}
\phi(i; z, \lambda, \hat{v}) 
&= \frac{z}{\lambda}\left( v_*^\top(X_i - \truemean) - \hat{v}^\top (X_i - \empmean )\right)
\end{align*}
By linearity of expectation: 
\begin{align*}
\E{i \sim W} \phi(i; z, \lambda, \hat{v}) 
&= \frac{z}{\lambda}v_*^\top(\weightedmean - \truemean) - \frac{z}{\lambda} \hat{v}^\top (\weightedmean - \empmean)
\end{align*}
Plugging these terms back into (\ref{donskervaradhan}) obtains, for all $W$: \begin{align*}
\KL(W\ ||\ U) \geq \sup_{z, \lambda, \hat{v}}
&\ \frac{z}{\lambda}v_*^\top(\weightedmean - \truemean)\ - \\
&\ \frac{z}{\lambda}\hat{v}^\top (\weightedmean - \empmean )\ - \\
&\ \log \E{i \sim U} \exp \phi(i; z, \lambda, \hat{v})
\end{align*}
This inequality holds no matter the choice of $z$, $\lambda$ and $\hat{v}$ in the supremum. Thus, following rearranged inequality holds for all (adaptively-chosen) $W$, $z$, $\lambda$ and $\hat{v}$:
\begin{align}
\label{basebayesinequality}
z v_*^\top(\weightedmean - \truemean) \leq
&\ 
\lambda\ \KL(W\ ||\ U)\ + \\
&\ z \hat{v}^\top (\weightedmean - \empmean )\ + \nonumber\\
&\ \lambda \log \E{i \sim U} \exp \phi(i; z, \lambda, \hat{v}) \nonumber
\end{align}
First, let's make sure that bounding the left-hand side leads to a bound on $|\hat{\mu}_1 - \mu_1|$. Recalling (\ref{linofmeans}), we just need the bound the right-hand side for both values of $z$:
\begin{align}
\label{maxoverz}
\max_{z \in \{-1,1\}}
z\ v_*^T( \weightedmean - \truemean) 
&= |\hat{\mu}_1 - \mu_1|
\end{align}
On the right-hand side of (\ref{basebayesinequality}), we will consider the worst-possible direction of $\hat{v}$ for the second term, so that we may choose it arbitrarily to minimize the third term. By the duality of the $\ell_1$ and $\ell_\infty$ norms, the second term is just the scaled imbalance: 
\begin{align*}
z \hat{v}^\top (\weightedmean - \widehat{M})
&= ||\hat{v}||_1 \frac{z \hat{v}^\top}{||\hat{v}||_1} (\weightedmean - \widehat{M}) \\
&\leq ||\hat{v}||_1 \sup_{||v||_1 \leq 1} v^\top (\weightedmean - \widehat{M}) \\
&= ||\hat{v}||_1 \max_{1 \leq j \leq d} \left| \sum_{\textrm{treated}\ i} W_i X_{ij} - \emptarget_j \right|
\end{align*} 
Let us pick $\hat{v}$ as follows: 
\begin{align*}
\hat{v} = (1-\delta) v_* + \delta u
\hspace{4mm}
\textrm{ where }
u = - z s \beta_z e_{j^*}
\end{align*}
By the triangle inequality and \Cref{asm:linearity}, we have $||\hat{v}||_1 \leq (1-\delta) k + \delta \beta_z$, so the imbalance is scaled to match (\ref{jointbalance}). Now we turn to bounding the \textrm{logmeanexp} moment term in (\ref{basebayesinequality}). Grouping terms in $\phi$ which depend on $i$:
\begin{align*}
\phi(i; z, \lambda, \hat{v}) 
&= \frac{z}{\lambda}\left( Y_i(1) - \hat{v}^\top X_i \right) +
\frac{z}{\lambda}\left(\hat{v}^\top \empmean - v_*^\top \truemean\right)
\end{align*}
The terms which don't depend on $i$ fall out of the exponent:
\begin{align*}
\lambda \log \E{i \sim U} \exp \phi
=&\ z (\hat{v}^\top \emptarget - v_*^\top \truemean)\ + \\
 &\  \lambda \log \E{i \sim U} \exp \frac{z}{\lambda}\left(Y_i(1)  - \hat{v}^\top X_i\right)
\end{align*}
The terms remaining in the exponent simplify as desired: \begin{align*}
\hat{v}^\top X_i &= (1-\delta) Y_i(1) + \delta z \widehat{Y}_i(1) \\
z(Y_i(1) - \hat{v}^\top X_i) &= \delta( z Y_i(1) - \widehat{Y}_i(1))
\end{align*}
In the linear terms that fell out, one part is bounded by H\"older's inequality and \Cref{asm:linearity,asm:concentration}:
\begin{align*}
z(\hat{v}^\top \emptarget - v_*^\top \truemean)
&= (1-\delta) z v_*^\top (\emptarget - \truemean) + z\delta (\ldots) \\
&\leq (1-\delta) C k \log(d) / \sqrt{n} + z\delta (\ldots)
\end{align*}
Finally, we show the other, elided $z\delta(\ldots)$ part is zero: 
\begin{align*}
z\delta (u^\top \emptarget - v_*^\top \truemean)
&= -\beta_z ||\emptarget||_\infty - z v_*^\top \truemean \\
&\leq -\beta_z ||\emptarget||_\infty + \beta_z ||\emptarget||_\infty = 0
\end{align*}
The first equality holds by construction of $u$. The second equality holds because the value of $\beta_z$ is given by (\ref{precursorprogram}). By \Cref{asm:linearity,asm:concentration},  that program's feasible set contains the actual $v_*$ and $\truemean$ with probability $1-\alpha$.

\section{Extensions of Main Analysis}

\subsection{Average Treatment Effects}
The following standard reduction is used to estimate ATE. This estimate still has confidence $1-\alpha$, because the concentration of \Cref{asm:concentration} has to be invoked just once.
\begin{theorem}
Suppose the treatment group weights $W$ attain value $\nu_1$ 
 in the program (\ref{jointbalance}). Let $\hat{\mu}_1$ be constructed as in (\ref{weightedestimate}). Analogously, suppose the control group weights $V$ attain value $\nu_0$ in the program (\ref{jointbalance}) when it involves reweighting the control $i$ rather than the treated $i$. Let $\hat{\mu}_0 = \sum_{\textrm{control}\ i} V_i Y_i(0)$ and $\widehat{\textrm{ATE}} = \hat{\mu}_1 - \hat{\mu}_0$. Then, under \cref{asm:linearity,asm:concentration}, with probability $1 - \alpha$,  $|\widehat{\textrm{ATE}} - \textrm{ATE}| \leq \nu_0 + \nu_1$.
\end{theorem}

\subsection{Asymmetric Confidence Interval}

Solving (\ref{jointbalance}) produces a single estimate $\hat{\mu}_1$ with a symmetric confidence interval. For an asymmetric (and potentially tighter) confidence interval, it is valid to solve (\ref{jointbalance}) twice, using fixed values of $z \in \{-1, 1\}$ each time. This achieves two different objective values $\nu_z$, and produces a pair of weights $W_z$ with two corresponding estimates $\hat{\mu}_{1,z}$. Then $[\hat{\mu}_{1,1} - \nu_1, \hat{\mu}_{1,-1} + \nu_{-1}]$ is a $1-\alpha$ confidence interval for $\mu_1$. This is because, for a one-sized bound, it is not necessary to take the maximum over $z$ in (\ref{maxoverz}). There is no loss in confidence because \Cref{asm:concentration} is invoked just once. 

\subsection{Consistency in High Dimensions}
\label{sec:consistency}

\begin{figure*}[ht]
  \centering
  \includegraphics[trim=0.1in 0.0in 0.1in 0.1in,clip,width=\textwidth]{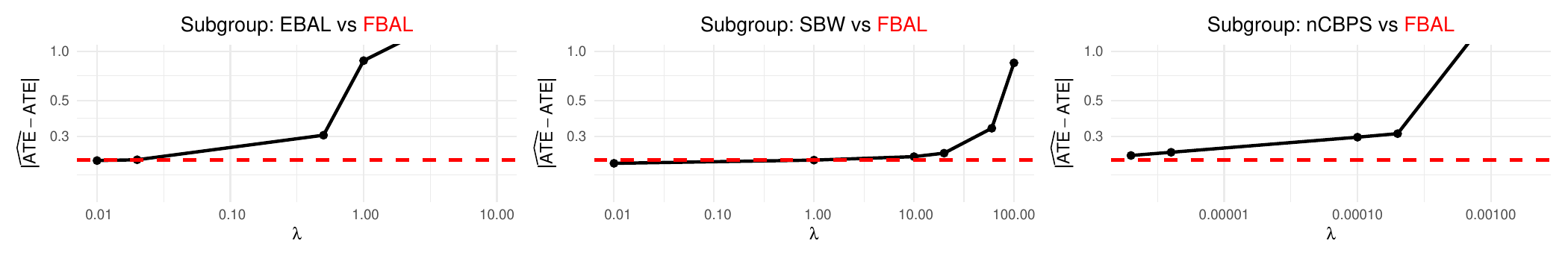}
  \includegraphics[trim=0.1in 0.1in 0.1in 0.1in,clip,width=\textwidth]{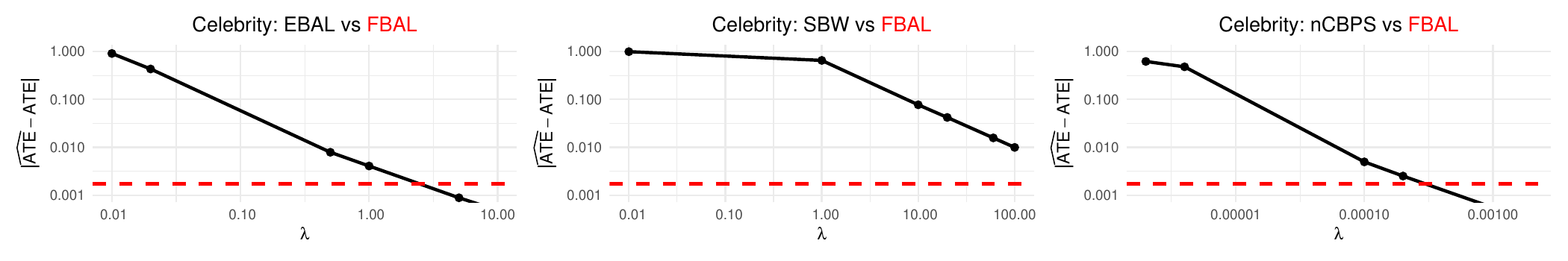}
  \vspace{-0.2in}
  \caption{Results of the \textbf{subgroup} simulation (above) and the \textbf{celebrity} simulation (below). Since the true $\textrm{ATE} = 0$ in both simulations, lower estimates are better. Flexible entropy balancing (FBAL) is compared to EBAL, SBW, and nCBPS in the left, middle, and right columns, respectively. For all of the competing methods, a fixed setting of $\lambda$ may do well in the subgroup simulation, but then that value performs badly in the celebrity simulation --- or vice versa. FBAL does not have a fixed hyperparameter, and is capable of solving both problems very well.}
  \label{fig:experimentresults}
\end{figure*}

The naive interval (\ref{naiveinterval}) shows that covariate balancing is consistent under \Cref{asm:linearity,asm:concentration,asm:optimalweights,asm:asymptotics}. Flexible covariate balancing inherits this property; see \Cref{appx:consistency} for proof. 
\begin{theorem}
\label{thm:consistency}
Let $\hat{\mu}_1$ be constructed as in (\ref{weightedestimate}) from the solution $W$ of (\ref{jointbalance}). Then, under \Cref{asm:linearity,asm:concentration,asm:asymptotics,asm:optimalweights}, $|\hat{\mu}_1 - \mu_1| \to 0$ as $n \to \infty$.
\end{theorem}

\section{Experiments}
\label{sec:experiments}

In this section, we compare flexible covariate balancing (FBAL, \ref{jointbalance}) to the baseline algorithms described in \Cref{sec:balancingalgs}. These competitors are abbreviated as EBAL, SBW, and nCBPS (where the n denotes weight normalization). Since these methods do not produce confidence intervals with provable coverage, we use (pointwise) estimation error as the primary evaluation metric.

\subsection{Simulations}
We present two purposefully-designed simulations in which algorithms driven by a fixed tolerance (or regularization) parameter will do poorly on at least one. It is not possible to simultaneously achieve good performance on both simulations without problem-specific, data-driven adaptivity.

\textbf{1. The subgroup simulation.} Suppose the combined population is 50\% white and 50\% black. The control outcomes $Y_i(0)$ are all zero. The treatment outcome is $Y_i(1) = 1$ in the white population, but is $-1$ in the black population, so $\textrm{ATE} = 0$. The treatment group is (exactly) 95\% white and 5\% black; the control group is (exactly) 5\% white and 95\% black. Besides race (which is encoded as $10$ for white, and $-10$ for black) the rest of the $d-1$ covariates are irrelevant $\textrm{Uniform}(-1,1)$. The correct approach is to heavily upweight the underrepresented black subgroup among the (mostly white) treated, and to similarly upweight the white subgroup among the (mostly black) controls. However, an algorithm with too high $\lambda$ (or $\tau$) will hesitate to place very heavy weights on just 5\% of the group. As $n$ increases, an algorithm must allow itself to diverge from uniform and perform this reweighting, balancing the covariate for race. 

\textbf{2. The celebrity simulation.} Suppose 95\% of the participants, in both the treatment and control groups, come from the general population. All $d$ of their covariates are $\textrm{Uniform}(-10,10)$, potentially exhibiting substantial individual variation. The remaining 5\% are celebrities, whose covariates are all completely zero vectors, perfectly matching the true mean. 
It is tempting for an algorithm to achieve low imbalance from the mean by placing heavy weight on the celebrities. Unfortunately, the celebrities are very atypical in the sense that, just for them, the treatment effect is substantially positive, equaling $1$. (This makes the simulation misspecified relative to \Cref{asm:linearity}). 
For the general population, the effect is mildly negative at $-0.05 / 0.95$, which renders the overall ATE null. As $n$ increases, the general population mean converges to zero, so the benefit of emphasizing the celebrities dissipates. The most accurate algorithm will keep uniform weights over the population to obtain a null estimate. An algorithm with too low $\lambda$ (or $\tau$) will not do this as well. 

So, heavily weighting the 5\% is beneficial in the first simulation, and it is harmful in the second. The algorithm is not allowed to be told, implicitly via problem-specific regularization parameters, which scenario is active. To be successful, it must judge different degrees of tradeoff between imbalance and divergence based upon the data. As \Cref{fig:experimentresults} illustrates, flexible covariate balancing is able to solve this challenge, because it is built from a more complete theory of how this tradeoff works.

\subsection{Real Trial Emulations}
\label{sec:realdataexperiments}

\begin{figure*}[ht]
  \centering
  \includegraphics[trim=0.1in 0.0in 0.1in 0.1in,clip,width=\textwidth]{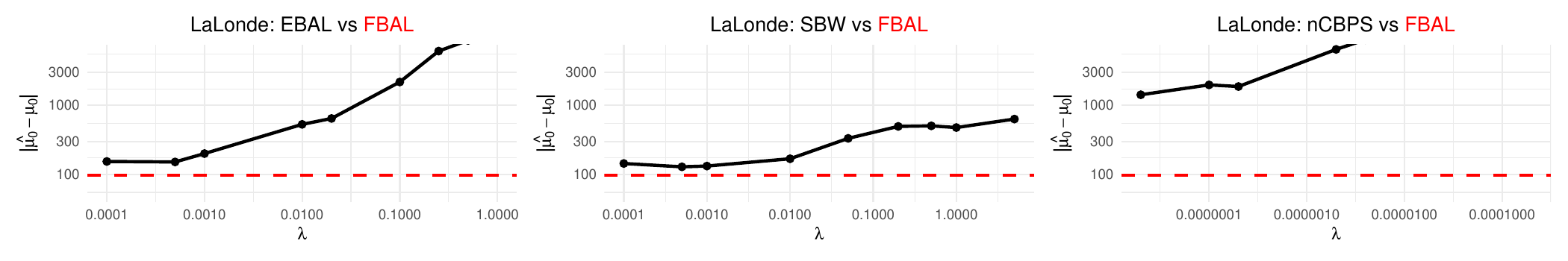}
  \includegraphics[trim=0.1in 0.1in 0.1in 0.1in,clip,width=\textwidth]{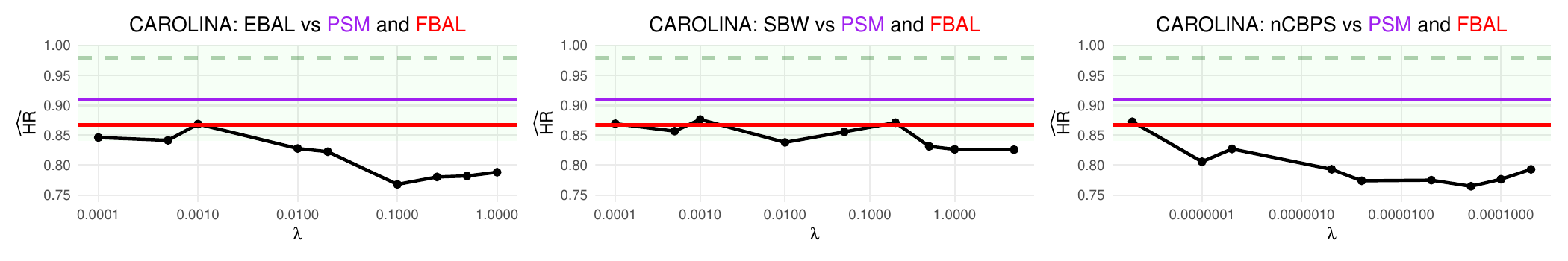}
  \vspace{-0.2in}
  \caption{Results of the \textbf{LaLonde} emulation (above) and the \textbf{CAROLINA} emulation (below). In the LaLonde plot, the error of the estimate $\hat{\mu}_0$ is on the $y$-axis, so lower values are better. For CAROLINA, higher estimates of HR are better: the trial's observed HR of 0.98 is plotted as a dashed green line, and its confidence interval of $[0.84-1.14]$ is shaded in light green. FBAL performs better than, or equal to, all the other covariate balancing algorithms. Propensity score matching (PSM) performs better than FBAL; however, it was performed on a separate dataset which is an order of magnitude larger, with careful manual oversight of the matching and covariate balance \citep{patorno2019using}. }
  \label{fig:realdataexperimentresults}
\end{figure*}

Both of the following experiments involve using observational data to emulate a randomized controlled trial. By treating an RCT as (approximate) ground truth, we avoid generating synthetic ground truth. That practice, though widespread in causal inference benchmarking, can be controversial and potentially misleading
\citep{curth2021really}.  

\textbf{1. Random-Fourier LaLonde.} This experiment is based on 
the dataset of \citet{lalonde1986evaluating}, which is a hobbyhorse in causal inference \citep{imbens2024lalonde}. The intervention in this dataset is a job training program, and the outcome is personal annual income. This dataset includes a randomized trial whose control group consists of $n = 260$ participants. By averaging their outcomes, we can reliably estimate the control response of \emph{not} participating in a job program, and informally treat it as ground truth: $\mu_0 = 4554.80$. The dataset also includes an observational control group with $n_0 = 15,992$ participants. Both the experimental and observational participants have $9$ covariates. To make this data more algorithmically involved, we expand $d = 16,000$ random Fourier features from the original $9$ \citep{rahimi2008weighted}. Then, we use different balancing algorithms to estimate $\mu_0$ (as $\hat{\mu}_0$) by reweighting the observational control.

The results are depicted in \Cref{fig:realdataexperimentresults}. As observed in previous works \citep{hainmueller2012entropy}, $\lambda \approx 0$ works well on this dataset; this setting is determined automatically by FBAL.

\textbf{2. The CAROLINA trial.} Between 2010 and 2018, \citet{ rosenstock2019effect} conducted a randomized controlled trial comparing linagliptin (a new drug) and glimepiride (a baseline drug) for patients with type 2 diabetes and elevated cardiovascular risk. Its primary outcome was the occurrence of a major cardiovascular event (myocardial infarction, nonfatal stroke, and/or cardiovascular death). From $n = 6033$ randomized patients, they estimated a hazard ratio (HR) of 0.98, with a 95\% confidence interval $[0.84, 1.14]$. Thus, the trial reasonably concluded that linagliptin's cardiovascular safety profile is noninferior to that of glimepiride.

While this trial was ongoing, \cite{patorno2019using} attempted to predict its result using observational data and propensity score matching. Based on the initially-published protocol and baseline statistics of the trial \citep{marx2015design}, \cite{patorno2019using} propensity-score matched 24,131 patients, for a total of 48,262 participants. They predicted $\widehat{\textrm{HR}} = 0.91$, which fell within the confidence interval eventually reported by the trial. Our goal is to do the same using covariate balancing algorithms, upon observational data drawn from a single institution within FeederNet \citep{you2022establishment}.

Our cohort for this emulation involves $n_0 = 3739$ patients given glimepiride and $n_1 = 488$ patients given linagliptin. There are 12,900 covariates, but most of them are very sparse; after removing covariates which are are over 50\% zero, $208$ covariates remain. (As shown in \Cref{fig:worsecarolina} in the Appendix, keeping these sparse features led to poor results from all the covariate balancing algorithms). From these 208 covariates, we extracted $d = 5000$ random Fourier features. In order to estimate HR, the learned weights are used to fit weighted Cox models upon the time-to-event data.

Our results are displayed in \Cref{fig:realdataexperimentresults}. As hoped, FBAL's estimate $\widehat{\textrm{HR}} = 0.867$ lies within the 95\% CI of the trial. By contrast, the other covariate balancing algorithms tended to predict outside the CI, depending on the setting of $\lambda$. However, this experiment may also indicate some limitations of covariate balancing as an algorithmic strategy. All the covariate balancing algorithms required some separate feature extraction, and they were still outperformed by propensity score matching. In covariate balancing, the maximum imbalance is fundamentally susceptible to the inclusion of irrelevant and/or corrupted features.

\section{Conclusion}

This paper initiates the PAC-Bayesian analysis of covariate balancing. It presents the first confidence interval which explains the generalization benefits of near-uniform weights in this setting. This theoretical result illuminates the tradeoff between imbalance and divergence, internalizes the complexities of covariate shift, offers an alternative proof of asymptotic consistency, and gives rise to a flexible new balancing algorithm with optimizable regularization parameters. This algorithm empirically demonstrates more adaptivity than existing methods. It would be interesting to pair our new balancing algorithm with a regression or propensity score model, much as \citet{athey2018approximate} pair stable balancing weights with the elastic net. 

Aside from causal inference, our techniques are of possible interest to the machine learning community in general. Rigorously, simultaneously optimizing both weights and hyperparameters is an interesting topic for further research. We achieved this through a uniform distribution on exchangeable data, and many linear learning problems can be dually expressed as (unnormalized, signed) weights over such data.  

\newpage
\bibliography{main}

\begin{thebibliography}{57}
\providecommand{\natexlab}[1]{#1}
\providecommand{\url}[1]{\texttt{#1}}
\expandafter\ifx\csname urlstyle\endcsname\relax
  \providecommand{\doi}[1]{doi: #1}\else
  \providecommand{\doi}{doi: \begingroup \urlstyle{rm}\Url}\fi

\bibitem[Alquier et~al.(2024)]{alquier2024user}
Pierre Alquier et~al.
\newblock User-friendly introduction to pac-bayes bounds.
\newblock \emph{Foundations and Trends{\textregistered} in Machine Learning}, 17\penalty0 (2):\penalty0 174--303, 2024.

\bibitem[Athey et~al.(2018)Athey, Imbens, and Wager]{athey2018approximate}
Susan Athey, Guido~W Imbens, and Stefan Wager.
\newblock Approximate residual balancing: debiased inference of average treatment effects in high dimensions.
\newblock \emph{Journal of the Royal Statistical Society Series B: Statistical Methodology}, 80\penalty0 (4):\penalty0 597--623, 2018.

\bibitem[Ben-Michael et~al.(2021)Ben-Michael, Feller, Hirshberg, and Zubizarreta]{ben2021balancing}
Eli Ben-Michael, Avi Feller, David~A Hirshberg, and Jos{\'e}~R Zubizarreta.
\newblock The balancing act in causal inference.
\newblock \emph{arXiv preprint arXiv:2110.14831}, 2021.

\bibitem[Bruns-Smith and Feller(2022)]{bruns2022outcome}
David~A Bruns-Smith and Avi Feller.
\newblock Outcome assumptions and duality theory for balancing weights.
\newblock In \emph{International Conference on Artificial Intelligence and Statistics}, pages 11037--11055. PMLR, 2022.

\bibitem[Catoni(2007)]{catoni2007pacbayes}
Olivier Catoni.
\newblock Pac-bayesian supervised classification: The thermodynamics of statistical learning.
\newblock \emph{Lecture Notes-Monograph Series}, 56:\penalty0 i--163, 2007.
\newblock ISSN 07492170.
\newblock URL \url{http://www.jstor.org/stable/20461499}.

\bibitem[Chan et~al.(2016)Chan, Yam, and Zhang]{chan2016globally}
Kwun Chuen~Gary Chan, Sheung Chi~Phillip Yam, and Zheng Zhang.
\newblock Globally efficient non-parametric inference of average treatment effects by empirical balancing calibration weighting.
\newblock \emph{Journal of the Royal Statistical Society Series B: Statistical Methodology}, 78\penalty0 (3):\penalty0 673--700, 2016.

\bibitem[Chattopadhyay et~al.(2020)Chattopadhyay, Hase, and Zubizarreta]{chattopadhyay2020balancing}
Ambarish Chattopadhyay, Christopher~H Hase, and Jos{\'e}~R Zubizarreta.
\newblock Balancing vs modeling approaches to weighting in practice.
\newblock \emph{Statistics in Medicine}, 39\penalty0 (24):\penalty0 3227--3254, 2020.

\bibitem[Chen et~al.(2021)Chen, Suchard, Krumholz, Schuemie, Shea, Duke, Pratt, Reich, Madigan, You, et~al.]{chen2021comparative}
RuiJun Chen, Marc~A Suchard, Harlan~M Krumholz, Martijn~J Schuemie, Steven Shea, Jon Duke, Nicole Pratt, Christian~G Reich, David Madigan, Seng~Chan You, et~al.
\newblock Comparative first-line effectiveness and safety of ace (angiotensin-converting enzyme) inhibitors and angiotensin receptor blockers: a multinational cohort study.
\newblock \emph{Hypertension}, 78\penalty0 (3):\penalty0 591--603, 2021.

\bibitem[Cousineau et~al.(2023)Cousineau, Verter, Murphy, and Pineau]{cousineau2023estimating}
Martin Cousineau, Vedat Verter, Susan~A Murphy, and Joelle Pineau.
\newblock Estimating causal effects with optimization-based methods: A review and empirical comparison.
\newblock \emph{European Journal of Operational Research}, 304\penalty0 (2):\penalty0 367--380, 2023.

\bibitem[Curth et~al.(2021)Curth, Svensson, Weatherall, and van~der Schaar]{curth2021really}
Alicia Curth, David Svensson, Jim Weatherall, and Mihaela van~der Schaar.
\newblock Really doing great at estimating {CATE}? a critical look at {ML} benchmarking practices in treatment effect estimation.
\newblock In \emph{Thirty-fifth Conference on Neural Information Processing Systems Datasets and Benchmarks Track (Round 2)}, 2021.
\newblock URL \url{https://openreview.net/forum?id=FQLzQqGEAH}.

\bibitem[Dahabreh et~al.(2020)Dahabreh, Robertson, Steingrimsson, Stuart, and Hernan]{dahabreh2020extending}
Issa~J Dahabreh, Sarah~E Robertson, Jon~A Steingrimsson, Elizabeth~A Stuart, and Miguel~A Hernan.
\newblock Extending inferences from a randomized trial to a new target population.
\newblock \emph{Statistics in medicine}, 39\penalty0 (14):\penalty0 1999--2014, 2020.

\bibitem[Dai et~al.(2020)Dai, Nachum, Chow, Li, Szepesv{\'a}ri, and Schuurmans]{dai2020coindice}
Bo~Dai, Ofir Nachum, Yinlam Chow, Lihong Li, Csaba Szepesv{\'a}ri, and Dale Schuurmans.
\newblock Coindice: Off-policy confidence interval estimation.
\newblock \emph{Advances in neural information processing systems}, 33:\penalty0 9398--9411, 2020.

\bibitem[Donsker and Varadhan(1983)]{donsker1983asymptotic}
Monroe~D Donsker and SR~Srinivasa Varadhan.
\newblock Asymptotic evaluation of certain markov process expectations for large time. iv.
\newblock \emph{Communications on pure and applied mathematics}, 36\penalty0 (2):\penalty0 183--212, 1983.

\bibitem[Fan et~al.(2016)Fan, Imai, Liu, Ning, Yang, et~al.]{fan2016improving}
Jianqing Fan, Kosuke Imai, Han Liu, Yang Ning, Xiaolin Yang, et~al.
\newblock Improving covariate balancing propensity score: A doubly robust and efficient approach.
\newblock \emph{URL: https://imai. fas. harvard. edu/research/CBPStheory. html}, 2016.

\bibitem[Hainmueller(2012)]{hainmueller2012entropy}
Jens Hainmueller.
\newblock Entropy balancing for causal effects: A multivariate reweighting method to produce balanced samples in observational studies.
\newblock \emph{Political analysis}, 20\penalty0 (1):\penalty0 25--46, 2012.

\bibitem[Hazlett(2020)]{hazlett2020kernel}
Chad Hazlett.
\newblock Kernel balancing.
\newblock \emph{Statistica Sinica}, 30\penalty0 (3):\penalty0 1155--1189, 2020.

\bibitem[Hazlett(2014)]{hazlett2014thesis}
Chad~J. Hazlett.
\newblock \emph{Inference in Tough Places: Essays on Modeling and Matching with Applications to Civil Conflict}.
\newblock Ph.d. thesis, Massachusetts Institute of Technology, 2014.
\newblock URL \url{http://hdl.handle.net/1721.1/92080}.

\bibitem[Hirshberg and Wager(2021)]{hirshberg2021augmented}
David~A Hirshberg and Stefan Wager.
\newblock Augmented minimax linear estimation.
\newblock \emph{The Annals of Statistics}, 49\penalty0 (6):\penalty0 3206--3227, 2021.

\bibitem[Hirshberg et~al.(2019)Hirshberg, Maleki, and Zubizarreta]{hirshberg2019minimax}
David~A Hirshberg, Arian Maleki, and Jose~R Zubizarreta.
\newblock Minimax linear estimation of the retargeted mean.
\newblock \emph{arXiv preprint arXiv:1901.10296}, 2019.

\bibitem[Hripcsak et~al.(2021)Hripcsak, Schuemie, Madigan, Ryan, and Suchard]{hripcsak2021drawing}
George Hripcsak, Martijn~J Schuemie, David Madigan, Patrick~B Ryan, and Marc~A Suchard.
\newblock Drawing reproducible conclusions from observational clinical data with ohdsi.
\newblock \emph{Yearbook of medical informatics}, 30\penalty0 (01):\penalty0 283--289, 2021.

\bibitem[Huling and Mak(2024)]{huling2024energy}
Jared~D Huling and Simon Mak.
\newblock Energy balancing of covariate distributions.
\newblock \emph{Journal of Causal Inference}, 12\penalty0 (1):\penalty0 20220029, 2024.

\bibitem[Imai and Ratkovic(2014)]{imai2014covariate}
Kosuke Imai and Marc Ratkovic.
\newblock Covariate balancing propensity score.
\newblock \emph{Journal of the Royal Statistical Society Series B: Statistical Methodology}, 76\penalty0 (1):\penalty0 243--263, 2014.

\bibitem[Imai et~al.(2010)Imai, Keele, and Yamamoto]{imai2010identification}
Kosuke Imai, Luke Keele, and Teppei Yamamoto.
\newblock Identification, inference and sensitivity analysis for causal mediation effects.
\newblock 2010.

\bibitem[Imbens and Xu(2024)]{imbens2024lalonde}
Guido Imbens and Yiqing Xu.
\newblock Lalonde (1986) after nearly four decades: Lessons learned.
\newblock \emph{arXiv preprint arXiv:2406.00827}, 2024.

\bibitem[Josey et~al.(2021{\natexlab{a}})Josey, Berkowitz, Ghosh, and Raghavan]{josey2021transporting}
Kevin~P Josey, Seth~A Berkowitz, Debashis Ghosh, and Sridharan Raghavan.
\newblock Transporting experimental results with entropy balancing.
\newblock \emph{Statistics in medicine}, 40\penalty0 (19):\penalty0 4310--4326, 2021{\natexlab{a}}.

\bibitem[Josey et~al.(2021{\natexlab{b}})Josey, Juarez-Colunga, Yang, and Ghosh]{josey2021framework}
Kevin~P Josey, Elizabeth Juarez-Colunga, Fan Yang, and Debashis Ghosh.
\newblock A framework for covariate balance using bregman distances.
\newblock \emph{Scandinavian Journal of Statistics}, 48\penalty0 (3):\penalty0 790--816, 2021{\natexlab{b}}.

\bibitem[Karampatziakis et~al.(2021)Karampatziakis, Mineiro, and Ramdas]{karampatziakis2021off}
Nikos Karampatziakis, Paul Mineiro, and Aaditya Ramdas.
\newblock Off-policy confidence sequences.
\newblock In \emph{International Conference on Machine Learning}, pages 5301--5310. PMLR, 2021.

\bibitem[Kaul and Gordon(2024)]{kaul2024meta}
Shiva Kaul and Geoffrey~J. Gordon.
\newblock Meta-analysis with untrusted data.
\newblock In \emph{Proceedings of the 4th Machine Learning for Health Symposium}, volume 259 of \emph{Proceedings of Machine Learning Research}, pages 562--592. PMLR, 2024.

\bibitem[Kong et~al.(2023)Kong, Park, Jung, Lee, and Kim]{kong2023covariate}
Insung Kong, Yuha Park, Joonhyuk Jung, Kwonsang Lee, and Yongdai Kim.
\newblock Covariate balancing using the integral probability metric for causal inference.
\newblock In \emph{International Conference on Machine Learning}, pages 17430--17461. PMLR, 2023.

\bibitem[Kuang et~al.(2017)Kuang, Cui, Li, Jiang, and Yang]{kuang2017estimating}
Kun Kuang, Peng Cui, Bo~Li, Meng Jiang, and Shiqiang Yang.
\newblock Estimating treatment effect in the wild via differentiated confounder balancing.
\newblock In \emph{Proceedings of the 23rd ACM SIGKDD international conference on knowledge discovery and data mining}, pages 265--274, 2017.

\bibitem[LaLonde(1986)]{lalonde1986evaluating}
Robert~J LaLonde.
\newblock Evaluating the econometric evaluations of training programs with experimental data.
\newblock \emph{The American economic review}, pages 604--620, 1986.

\bibitem[Marx et~al.(2015)Marx, Rosenstock, Kahn, Zinman, Kastelein, Lachin, Espeland, Bluhmki, Mattheus, Ryckaert, et~al.]{marx2015design}
Nikolaus Marx, Julio Rosenstock, Steven~E Kahn, Bernard Zinman, John~J Kastelein, John~M Lachin, Mark~A Espeland, Erich Bluhmki, Michaela Mattheus, Bart Ryckaert, et~al.
\newblock Design and baseline characteristics of the cardiovascular outcome trial of linagliptin versus glimepiride in type 2 diabetes (carolina{\textregistered}).
\newblock \emph{Diabetes and Vascular Disease Research}, 12\penalty0 (3):\penalty0 164--174, 2015.

\bibitem[McAllester(1998)]{mcallester1998some}
David~A McAllester.
\newblock Some pac-bayesian theorems.
\newblock In \emph{Proceedings of the eleventh annual conference on Computational learning theory}, pages 230--234, 1998.

\bibitem[Ning et~al.(2020)Ning, Sida, and Imai]{ning2020robust}
Yang Ning, Peng Sida, and Kosuke Imai.
\newblock Robust estimation of causal effects via a high-dimensional covariate balancing propensity score.
\newblock \emph{Biometrika}, 107\penalty0 (3):\penalty0 533--554, 2020.

\bibitem[Ohnishi and Honorio(2021)]{ohnishi2021novel}
Yuki Ohnishi and Jean Honorio.
\newblock Novel change of measure inequalities with applications to pac-bayesian bounds and monte carlo estimation.
\newblock In \emph{International conference on artificial intelligence and statistics}, pages 1711--1719. PMLR, 2021.

\bibitem[Patorno et~al.(2019)Patorno, Schneeweiss, Gopalakrishnan, Martin, and Franklin]{patorno2019using}
Elisabetta Patorno, Sebastian Schneeweiss, Chandrasekar Gopalakrishnan, David Martin, and Jessica~M Franklin.
\newblock Using real-world data to predict findings of an ongoing phase iv cardiovascular outcome trial: cardiovascular safety of linagliptin versus glimepiride.
\newblock \emph{Diabetes care}, 42\penalty0 (12):\penalty0 2204--2210, 2019.

\bibitem[Rahimi and Recht(2008)]{rahimi2008weighted}
Ali Rahimi and Benjamin Recht.
\newblock Weighted sums of random kitchen sinks: Replacing minimization with randomization in learning.
\newblock \emph{Advances in neural information processing systems}, 21, 2008.

\bibitem[Rodr{\'\i}guez-G{\'a}lvez et~al.(2024)Rodr{\'\i}guez-G{\'a}lvez, Thobaben, and Skoglund]{rodriguez2024more}
Borja Rodr{\'\i}guez-G{\'a}lvez, Ragnar Thobaben, and Mikael Skoglund.
\newblock More pac-bayes bounds: From bounded losses, to losses with general tail behaviors, to anytime validity.
\newblock \emph{Journal of Machine Learning Research}, 25\penalty0 (110):\penalty0 1--43, 2024.

\bibitem[Rosenbaum and Rubin(1983)]{rosenbaum1983central}
Paul~R Rosenbaum and Donald~B Rubin.
\newblock The central role of the propensity score in observational studies for causal effects.
\newblock \emph{Biometrika}, 70\penalty0 (1):\penalty0 41--55, 1983.

\bibitem[Rosenstock et~al.(2019)Rosenstock, Kahn, Johansen, Zinman, Espeland, Woerle, Pfarr, Keller, Mattheus, Baanstra, et~al.]{rosenstock2019effect}
Julio Rosenstock, Steven~E Kahn, Odd~Erik Johansen, Bernard Zinman, Mark~A Espeland, Hans~J Woerle, Egon Pfarr, Annett Keller, Michaela Mattheus, David Baanstra, et~al.
\newblock Effect of linagliptin vs glimepiride on major adverse cardiovascular outcomes in patients with type 2 diabetes: the carolina randomized clinical trial.
\newblock \emph{Jama}, 322\penalty0 (12):\penalty0 1155--1166, 2019.

\bibitem[Rubin(2008)]{rubin2008objective}
Donald~B. Rubin.
\newblock {For objective causal inference, design trumps analysis}.
\newblock \emph{The Annals of Applied Statistics}, 2\penalty0 (3):\penalty0 808 -- 840, 2008.
\newblock \doi{10.1214/08-AOAS187}.
\newblock URL \url{https://doi.org/10.1214/08-AOAS187}.

\bibitem[Ruderman et~al.(2012)Ruderman, Reid, Garc{\'\i}a-Garc{\'\i}a, and Petterson]{ruderman2012tighter}
Avraham Ruderman, Mark Reid, Dar{\'\i}o Garc{\'\i}a-Garc{\'\i}a, and James Petterson.
\newblock Tighter variational representations of f-divergences via restriction to probability measures.
\newblock \emph{arXiv preprint arXiv:1206.4664}, 2012.

\bibitem[Schuemie et~al.(2020)Schuemie, Cepeda, Suchard, Yang, Tian, Schuler, Ryan, Madigan, and Hripcsak]{schuemie2020confident}
Martijn~J Schuemie, M~Soledad Cepeda, Marc~A Suchard, Jianxiao Yang, Yuxi Tian, Alejandro Schuler, Patrick~B Ryan, David Madigan, and George Hripcsak.
\newblock How confident are we about observational findings in healthcare: a benchmark study.
\newblock \emph{Harvard data science review}, 2\penalty0 (1), 2020.

\bibitem[Shawe-Taylor and Williamson(1997)]{shawe1997pac}
John Shawe-Taylor and Robert~C Williamson.
\newblock A pac analysis of a bayesian estimator.
\newblock In \emph{Proceedings of the tenth annual conference on Computational learning theory}, pages 2--9, 1997.

\bibitem[S{\l}oczy{\'n}ski et~al.(2023)S{\l}oczy{\'n}ski, Uysal, and Wooldridge]{sloczynski2023covariate}
Tymon S{\l}oczy{\'n}ski, S~Derya Uysal, and Jeffrey~M Wooldridge.
\newblock Covariate balancing and the equivalence of weighting and doubly robust estimators of average treatment effects.
\newblock \emph{arXiv preprint arXiv:2310.18563}, 2023.

\bibitem[Su et~al.(2023)Su, Mou, Ding, and Wainwright]{su2023estimated}
Fangzhou Su, Wenlong Mou, Peng Ding, and Martin~J Wainwright.
\newblock When is the estimated propensity score better? high-dimensional analysis and bias correction.
\newblock \emph{arXiv preprint arXiv:2303.17102}, 2023.

\bibitem[Suchard et~al.(2019)Suchard, Schuemie, Krumholz, You, Chen, Pratt, Reich, Duke, Madigan, Hripcsak, et~al.]{suchard2019comprehensive}
Marc~A Suchard, Martijn~J Schuemie, Harlan~M Krumholz, Seng~Chan You, RuiJun Chen, Nicole Pratt, Christian~G Reich, Jon Duke, David Madigan, George Hripcsak, et~al.
\newblock Comprehensive comparative effectiveness and safety of first-line antihypertensive drug classes: a systematic, multinational, large-scale analysis.
\newblock \emph{The Lancet}, 394\penalty0 (10211):\penalty0 1816--1826, 2019.

\bibitem[Tian et~al.(2018)Tian, Schuemie, and Suchard]{tian2018evaluating}
Yuxi Tian, Martijn~J Schuemie, and Marc~A Suchard.
\newblock Evaluating large-scale propensity score performance through real-world and synthetic data experiments.
\newblock \emph{International journal of epidemiology}, 47\penalty0 (6):\penalty0 2005--2014, 2018.

\bibitem[Tibshirani(1996)]{tibshirani1996regression}
Robert Tibshirani.
\newblock Regression shrinkage and selection via the lasso.
\newblock \emph{Journal of the Royal Statistical Society Series B: Statistical Methodology}, 58\penalty0 (1):\penalty0 267--288, 1996.

\bibitem[Uehara et~al.(2022)Uehara, Shi, and Kallus]{uehara2022review}
Masatoshi Uehara, Chengchun Shi, and Nathan Kallus.
\newblock A review of off-policy evaluation in reinforcement learning.
\newblock \emph{arXiv preprint arXiv:2212.06355}, 2022.

\bibitem[Wang and Zubizarreta(2020)]{wang2020minimal}
Yixin Wang and Jose~R Zubizarreta.
\newblock Minimal dispersion approximately balancing weights: asymptotic properties and practical considerations.
\newblock \emph{Biometrika}, 107\penalty0 (1):\penalty0 93--105, 2020.

\bibitem[Waudby-Smith et~al.(2024)Waudby-Smith, Wu, Ramdas, Karampatziakis, and Mineiro]{waudby2024anytime}
Ian Waudby-Smith, Lili Wu, Aaditya Ramdas, Nikos Karampatziakis, and Paul Mineiro.
\newblock Anytime-valid off-policy inference for contextual bandits.
\newblock \emph{ACM/JMS Journal of Data Science}, 1\penalty0 (3):\penalty0 1--42, 2024.

\bibitem[Wong and Chan(2018)]{wong2018kernel}
Raymond~KW Wong and Kwun Chuen~Gary Chan.
\newblock Kernel-based covariate functional balancing for observational studies.
\newblock \emph{Biometrika}, 105\penalty0 (1):\penalty0 199--213, 2018.

\bibitem[You et~al.(2022)You, Lee, Choi, and Park]{you2022establishment}
Seng~Chan You, Seongwon Lee, Byungjin Choi, and Rae~Woong Park.
\newblock Establishment of an international evidence sharing network through common data model for cardiovascular research.
\newblock \emph{Korean Circulation Journal}, 52\penalty0 (12):\penalty0 853--864, 2022.

\bibitem[Zhang et~al.(2022)Zhang, Wang, Schuemie, Blei, and Hripcsak]{zhang2022adjusting}
Linying Zhang, Yixin Wang, Martijn~J Schuemie, David~M Blei, and George Hripcsak.
\newblock Adjusting for indirectly measured confounding using large-scale propensity score.
\newblock \emph{Journal of biomedical informatics}, 134:\penalty0 104204, 2022.

\bibitem[Zhao and Percival(2017)]{zhao2017entropy}
Qingyuan Zhao and Daniel Percival.
\newblock Entropy balancing is doubly robust, 2017.

\bibitem[Zubizarreta(2015)]{zubizarreta2015stable}
Jos{\'e}~R Zubizarreta.
\newblock Stable weights that balance covariates for estimation with incomplete outcome data.
\newblock \emph{Journal of the American Statistical Association}, 110\penalty0 (511):\penalty0 910--922, 2015.

\end{thebibliography}

\newpage
\clearpage
\appendix

\begin{figure*}[t!]
  \centering
  \includegraphics[trim=0.1in 2.2in 0.1in 2.2in,clip,width=\textwidth]{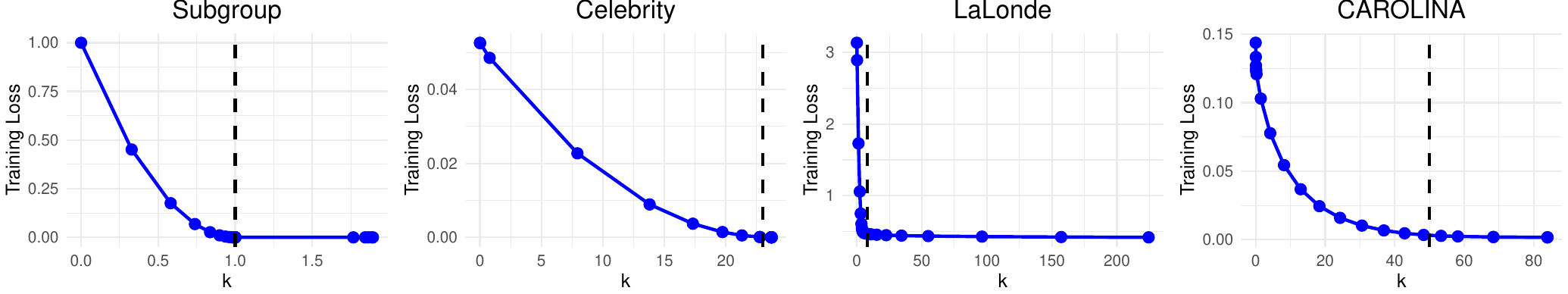}
  \caption{Residual bend analyses used to choose $k$ in the experiments. For each experiment, the observed data $(X, Y)$ are fit by the Lasso \citep{tibshirani1996regression}, for different regularization values $\lambda$. $k$ is chosen as the $\ell_1$ norm where the training loss begins to depart sharply from zero. This is a somewhat imprecise heuristic, but seems to work reasonably in practice.}
  \label{fig:residualbends}
\end{figure*}

\begin{figure*}[t!]
  \centering
  \includegraphics[trim=0.1in 0.1in 0.1in 0.1in,clip,width=\textwidth]{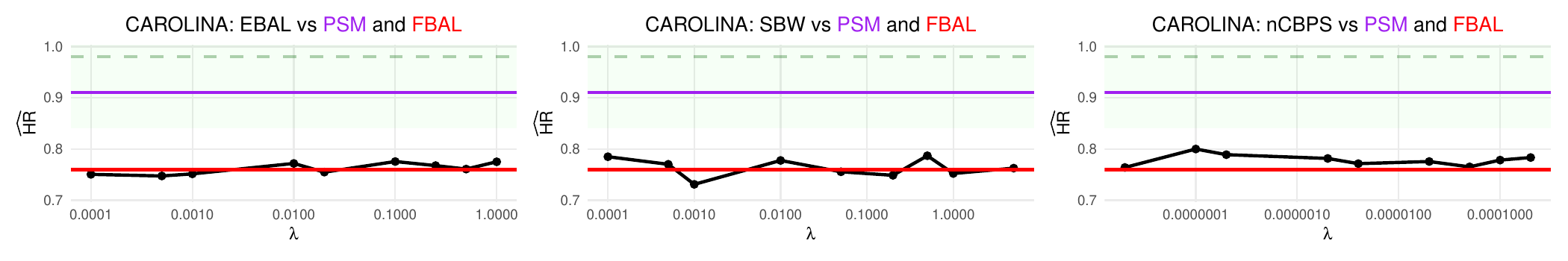}
  \vspace{-0.2in}
  \caption{Worse results for the \textbf{CAROLINA} emulation (below) when sparse features are not removed. A naive, unweighted Cox model yields $\widehat{HR} = 0.793$, so all of the balancing methods have essentially trivial performance. This makes intuitive sense: covariate balancing methods must balance \emph{all} the covariates, so including noisy or meaningless covariates greatly impacts their performance. By contrast, irrelevant features can be automatically ignored in propensity score estimation --- though they are still (manually) inspected and adjudicated in matching. }
  \label{fig:worsecarolina}
\end{figure*}

\section{$f$-divergence for CBPS}
\label{appx:cbps}

In this section, we present the $f$-divergence which places the covariate balancing propensity score \citep{imai2014covariate} in the form of (\ref{abstractbalancing}).
\begin{lemma}[CBPS]
\label{lem:cbpsdivergence}
Suppose there is a constant covariate (i.e. an index $j$ such that $X_{ij} = \emptarget_j$ for all $i$). Then CBPS, when formulated to estimate $\mu_1$, is equivalent to (\ref{abstractbalancing}) with $\tau = 0$ and the following $f$-divergence away from uniform: \begin{align*}
W \mapsto \frac{1}{n_1} \sum_{\textrm{treated}\ i} \Big(\frac{1}{W_i}  - 1\Big) \log\Big(\frac{1}{W_i} - 1\Big) - \frac{1}{W_i} + c
\end{align*}
In this divergence, the constant $c$ is defined in (\ref{cforcbps}).
\end{lemma}
CBPS fits a logistic regression to the propensity scores, but constrains the fitted scores to satisfy covariate balancing constraints. When estimating $\mu_1$, these constraints are:
\begin{align*}
\sum_{\textrm{treated}\ i} \frac{1}{p_i} X_{ij} = \emptarget_j \quad \forall\ 1 \leq j \leq d
\end{align*}
Here, $p_i \in [1, \infty)$ is the fitted propensity score given $X_i$. As highlighted by \citet{sloczynski2023covariate}, if there is a constant covariate, this normalizes the inverse-propensity scores: \[
\sum_{\textrm{treated}\ i} \frac{1}{p_i} = 1
\]
As described by \citet{josey2021framework}, the CBPS primal objective can be rephrased to minimize over the $p_i$ themselves, rather than the logistic regression coefficients: \begin{align*}
\textrm{CBPS}(p) &= \sum_{\textrm{treated}\ i} (p_i - 1) \log(p_i - 1) - p_i + 2
\end{align*}
(Their displayed formula actually includes target weights $q_i$, but these are set uniformly to $q_i = 2$). Since (\ref{abstractbalancing}) optimizes over probability weights, we perform the inversion $W_i = 1 / p_i$. Note that the two aforementioned constraints on the $p_i$ correctly translate to $W$ being a probability distribution. 
In analogy to the CBPS objective, let us define the $f$-divergence away from the uniform distribution using the following function:
\[
f(t) = \left(\frac{n_1}{t} - 1\right) \log\left(\frac{n_1}{t} - 1\right) - \frac{n_1}{t} + c
\]  
Here, the constant $c$ is chosen to make $f(1) = 0$: 
\begin{align}
\label{cforcbps}
c = -(n_1 - 1)\log(n_1 - 1) + n_1
\end{align}
It is easy to see that the divergence (\ref{fdivergence}) matches the CBPS objective, up to constants:
\begin{align*}
\textrm{Divergence}(W\ ||\ U)
&= \sum_{\textrm{treated}\ i} \frac{1}{n_1} f(n_1 \cdot W_i) \\
&= \frac{1}{n_1} \sum_{\textrm{treated}\ i}  f\left(\frac{n_1}{p_i}\right) 
\propto \textrm{CBPS}(p)
\end{align*}
We simply need to show that this $f$ satisfies the necessary conditions to define an $f$-divergence, per \Cref{sec:notation}. As mentioned before, $f(1) = 0$. Next, note $\lim_{t\to 0^+} f(t)$ does indeed grow to $f(0) = \infty$. Finally, we need to show $f$ is convex. When used to define a divergence away from the uniform distribution, $f$ is evaluated only over $t \in [0,n_1]$. Thus, for \Cref{lem:cbpsdivergence}, $f$ needs to be convex just over that range. Its second derivative is:
\begin{align*}
f''(t) 
&= \frac{n_1}{t^3} \left( 2 \log\left(\frac{n_1 - t}{t}\right) + \frac{n_1}{n_1 - t} \right) \\
&= \frac{n_1}{t^3} \underbrace{\left( 2 \log\left(\frac{1 - s}{s}\right) + \frac{1}{1 - s} \right)}_{g(s)} 
\end{align*}
The second line substitutes \( s = t / n_1 \). for $s \in [0,1]$. To prove $f$ is convex, we just need to show $g(s) \geq 0$ on $[0,1]$. As \( s \to 0 \) or \( s \to 1 \), \( g(s) \to +\infty \). Now we show $g$ is nonnegative within those endpoints by finding its sole critical point, showing it is a local minimum, and that $g$ remains nonnegative there. Taking derivatives: 
\begin{align*}
g'(s) &= \frac{-2}{s(1 - s)} + \frac{1}{(1 - s)^2} \\
g''(s) &= \frac{2(1 - 2s)}{s^2(1 - s)^2} + \frac{2}{(1 - s)^3}
\end{align*}
Solving $g'(s) = 0$, the unique critical point is at $s_{\textrm{min}} = 2/3$. We find $g''(s_{\textrm{min}}) > 0$ and $g(s_{\textrm{min}}) > 0$. This completes the proof of convexity.

\section{Additional Related Work}
\label{appx:relatedwork}

\subsection{Modeling Algorithms}

Another approach to causal inference is modeling the propensity score. This is the probability $e(x) = \mathbb{P}(T = 1\ |\ X = x)$ that a participant having covariates $x$ is assigned to the treatment group. Such modeling is often performed using $\ell_1$-regularized logistic regression \citep{tibshirani1996regression}. Though the modeling approach is superficially distinct from covariate balancing, they are mathematically intertwined \citep{ben2021balancing}. The inverse propensity score produces weights $W^*_i = 1 / e(X_i)$ which balance all covariate means. In fact, these are the unique weights which balance means of all bounded functions of the covariates. In terms of optimization, balancing and modeling are linked via convex duality. 

Due to these close connections, some algorithms attempt to simultaneously model and balance. As previously discussed, CBPS \citep{imai2014covariate} fits a logistic regression propensity score while constraining the inverse propensity scores (i.e. the weights) to balance all covariate means. 
Its successor HDCBPS \citep{ning2020robust} does not attempt to balance all covariates, but focuses on balancing only covariates which are relevant to the outcomes. This is akin to how Differentiated Confounder Balancing \citep{kuang2017estimating} learns which subset of covariates are confounders along with the weights $W$ balancing that subset. 

Matching \citep{rosenbaum1983central} is another way to use the propensity score. This pairs each member of the treatment group with a control who has approximately the same propensity score; otherwise, if no such control is found, the member is dropped. Because this method may not use all of the data, it can be problematic when $n$ is relatively small. Furthermore, if the minor differences between matched pairs accumulate in the same direction, substantial bias can arise; overall covariate imbalance still needs to be separately adjudicated. 

\subsection{Reinforcement Learning}

The observational study setup in this paper is a special case of off-policy evaluation in reinforcement learning \citep{uehara2022review}. More specifically, it is a special case of behavior-agnostic off-policy evaluation in contextual bandits. In our setting, the states are covariates, the two actions are assignment to either treatment or control, the historical policy is the unknown selection mechanism (i.e. the propensity score) in the observational data, and ATE is the value of the policy which chooses assignment randomly. In this more difficult and general setting, finite-sample confidence intervals \citep{dai2020coindice} have been derived. However, it is unclear whether they improve the naive interval (\ref{naiveinterval}) when specialized to observational studies. Subsequent works on off-policy confidence sequences \citep{karampatziakis2021off, waudby2024anytime} assume a known behavior policy (i.e. propensity score), which makes them inapplicable to our setting. 

\section{Existence of $W^*$}
\label{appx:existenceofipw}

This section describes why \Cref{asm:optimalweights} is a consequence of standard assumptions in causal inference. The two following assumptions are jointly known as strong ignorability \citep{rosenbaum1983central}. 
\begin{assumption}[Unconfoundedness]
\label{asm:unconfounded} For all $i \in \{1, \ldots, n\}$, 
$T_i \perp Y_i(0), Y_i(1)\ |\ X_i$. (Given the covariates, the treatment selection is independent of both potential outcomes). 
\end{assumption}
\begin{assumption}[Overlap]
\label{asm:overlap}
There is a constant $\gamma > 0$ such that, for all $x$, $\gamma < \mathbb{P}(T_i = 1\ |\ X = x) < 1 - \gamma$. (Each participant has a positive probability of being assigned to either the treatment or the control).
\end{assumption}
For the moment, let us assume $\mathcal{S} = \mathcal{T}$ (i.e. the source and target distributions are the same). Define $W^*$ as the true (unknown) inverse-propensity weights: \begin{align}
\label{inversepropensityweights}
W^*_i \propto \frac{1}{\mathbb{P}(T_i = 1\ |\ X = X_i)}
\hspace{5mm}
\sum_{\textrm{treated}\ i} W^*_i = 1
\end{align}
Under \Cref{asm:unconfounded,asm:overlap}, it is a standard result that $\sum_{\textrm{treated}\ i} W^*_i Y_i(1)$ has mean $M$ (see e.g. \citet{ben2021balancing}). Thus, by the central limit theorem, its error is $O(1/\sqrt{n})$. Next, we show $W^*$ satisfies the ratio-boundedness property of \Cref{asm:optimalweights}.
\begin{lemma}
Under \Cref{asm:overlap}:
\[
R
= 
\frac{\max_i W^*_i}{\min_i W^*_i}
\;\le\;
\biggl(\frac{1-\gamma}{\gamma}\biggr)^2.
\]
\end{lemma}
\begin{proof}
Abbreviate $W = W_*$, $n = n_1$ and $p_i = \mathbb{P}(T_i = 1\ |\ X = X_i)$, so $W_i = \frac{p_i}{Z}$ where $Z = \sum_i \frac{1}{p_i}$. By \Cref{asm:overlap}, each \(p_i \in (\gamma, 1-\gamma)\). Thus, we have:
\[
\frac{1}{1-\gamma} \;\le\; \frac{1}{p_i} \;\le\; \frac{1}{\gamma}.
\]
Therefore, taking sums:
\[
\frac{n}{1-\gamma} 
\;\le\; 
Z 
\;\le\; 
\frac{n}{\gamma}.
\]
Hence each component of \(W\) lies within this range: 
\begin{align*}
W_i 
\;=\;
\frac{1}{p_i\,Z}
\;\in\;
&
\biggl[\,
\frac{1}{\bigl(\tfrac{1}{\gamma}\bigr)\,\bigl(\tfrac{n}{\gamma}\bigr)},
\;\;
\frac{1}{\bigl(\tfrac{1}{1-\gamma}\bigr)\,\bigl(\tfrac{n}{1-\gamma}\bigr)}
\biggr] \\
\;=\;
&
\biggl[\,
\frac{\gamma}{n(1-\gamma)},
\;\;
\frac{\,1-\gamma\,}{n\,\gamma}
\biggr]
\end{align*}
Dividing these by one another, we see the ratio of largest to smallest $W_i$ is bounded.
\end{proof}
Now let us consider the case when $\mathcal{S} \neq \mathcal{T}$. In the literature studying this case, it is common to assume a positive participation probability, i.e. that all participants in $\mathcal{T}$ have some chance of being sampled in $\mathcal{S}$. To ensure ratio-boundedness of $W^*$, we can make the following, somewhat stronger assumption.
\begin{assumption}[Bounded density ratio]
There are constants $r_{0}, r_\infty > 0$ such that $r_0 \leq \frac{d\mathcal{S}(x)}{d \mathcal{T}(x)} \leq r_\infty$. 
\end{assumption}
Then, $W^*$ just needs to be adjusted by this ratio, and the standard arguments go through as before. 


\section{Linear Program for $\beta_z$}
\label{appx:linearprogram}

This section shows that (\ref{precursorprogram}) can be reformulated as a linear program. (Strictly speaking, this linear program is an upper bound, but it will usually be exactly tight). For readability, we copy (\ref{precursorprogram}) here: 
\begin{align*}
\max_{v_*, \truemean \in \mathbb{R}^d} & \quad \frac{-z}{||\emptarget||_\infty} v_*^\top M \\
\text{s.t.} & \quad ||v_*||_1 \leq k \hspace{6mm} ||\emptarget - M||_\infty \leq \frac{C \log d}{\sqrt{n}} \nonumber\\
            & \quad v_*^\top X_i = Y_i(1) \hspace{2mm} \textrm{for all treated}\ i \nonumber
\end{align*}
The first step is to eliminate the optimization over $\truemean$. For any fixed \(v_*\), the objective is maximized at \( \truemean = \emptarget - z \frac{C \log d}{\sqrt{n}} \text{sign}(v_*) \). (For example, in case $z = -1$, this perturbs $\emptarget$ in the direction of $\text{sign}(v_*)$. Using this optimal $M$, along with the the fact that \( v_*^\top \text{sign}(v_*) = \|v_*\|_1 \), the objective becomes just a maximization over $v_*$:  
\[
\frac{1}{||\emptarget||_\infty}\left( -z v_*^\top \widehat{M} + \frac{C \log d}{\sqrt{n}} \|v_*\|_1 \right)
\]
The $\ell_1$ norm is not concave, but this is not a problem. The constraint $||v_*||_1 \leq k$ is typically active (i.e. the sparsity constraint makes it more challenging to fit the data). This is trivially the case if there is an all-zero covariate. Thus, we can replace the non-concave $||v_*1||$ by just $k$. 

Now the $\ell_1$ norm constraint is linearized in the usual fashion. Let \( v_* = u - w \) with \( u, w \geq 0 \). Then \( \|v_*\|_1 = \sum_{i=1}^d (u_i + w_i)\) and  \( v_*^\top X_i = (u - w)^\top X_i\). The final program is:
\begin{align*}
\max_{u,v \in \mathbb{R}^d} & \quad \frac{1}{||\emptarget||_\infty}\left( \frac{C k \log d}{\sqrt{n}} - z \sum_{j=1}^d \emptarget_j (u_j - w_j)  \right) \\
\text{s.t.} & \quad \sum_{j=1}^d u_j + w_j \leq k \hspace{6mm} u_j,v_j \geq 0 \nonumber\\
            & \quad (u-w)^\top X_i = Y_i(1) \hspace{2mm} \textrm{for all treated}\ i \nonumber
\end{align*}
The potential for $||\emptarget||_\infty \approx 0$ may make this program seem ominous. However, in the linear setting of \Cref{asm:linearity}, is typically the case that a constant covariate is added. It is appropriate to think of $||\emptarget||_\infty = 1$ in our setting.

\section{General $f$-Divergences}
\label{appx:otherfdivergences}

This section explains how to extend the proof of \Cref{sec:proofwithkl} to all $f$-divergences, thereby establishing \Cref{thm:confidenceinterval} in full generality. This elaboration is very straightforward. Instead of beginning with the Donsker-Varadhan formula (\ref{donskervaradhan}), we use the more abstract variational representation of $f$-divergences \citep{ruderman2012tighter}: \begin{align*}
\textrm{Divergence}(W\ ||\ U)
= \sup_\phi & \E{i \sim W} \phi(i)\ - \\
            & \textrm{Divergence}^*([\phi(i)]_{\textrm{treated}\ i}\ ||\ U)
\end{align*}
The proof remains the same until the part where terms which don't depend on $i$ drop out of the \textrm{logmeanexp}. It is easy to see that the same equality holds for $\textrm{Divergence}^*$. Let $z_0 \in \mathbb{R}$ be an additive term constant across all the $i$. Because $W$ is a probability distribution:
\begin{align*}
 &\ \textrm{Divergence}^*(Z + \mathbf{1} z_0\ ||\ U) \\
=&\ \sup_{W} (Z + \mathbf{1} z_0)^\top W - \textrm{Divergence}(W\ ||\ U) \\
=&\ \sup_{W} z_0 \mathbf{1}^\top W + Z^\top W - \textrm{Divergence}(W\ ||\ U) \\
=&\ z_0 + \textrm{Divergence}^*(Z\ ||\ U)
\end{align*}
This completes the generalization of the proof. Let us see how it can be used to derive another covariate balancing algorithm. Consider the $\chi^2$ divergence and its conjugate \citep{ohnishi2021novel}: 
\begin{align}
\label{sfbal}
\chi^2(W\ ||\ U)
&= \sum_{\textrm{treated } i} \Big(W_i - \frac{1}{n_1}\Big)^2 \\
\chi^{2*}(Z\ ||\ U)
&= \E{i\sim U} Z_i - \frac{1}{4}\Var{i \sim U}\ Z_i \nonumber 
\end{align}
Using these in (\ref{jointbalance}) yields a flexible analogue of stable balancing weights \citep{zubizarreta2015stable}. As discussed in \Cref{sec:proofwithkl}, the data can be centered, without loss of generality, to eliminate the mean term. This instantiation of (\ref{jointbalance}) makes it clear that $\textrm{Divergence}^*$ quantifies empirical variation of the data. 

\section{Proof of Consistency}
\label{appx:consistency}

\Cref{asm:linearity,asm:concentration} secure the finite-sample guarantee in \Cref{thm:confidenceinterval}. Take $W = W^*$, $\delta = 0$ and $\lambda \to 0$. (The actual solution to (\ref{jointklbalance}) may be better, but this only helps our argument). Setting $\lambda \to 0$ eliminates the divergence terms. The remaining terms are $O(k \log(d) / \sqrt{n})$ which, by \Cref{asm:optimalweights,asm:asymptotics}, goes to zero.

\section{Experiment Details}

\textbf{Algorithms}. In most of the experiments, FBAL is flexible entropy balancing (\ref{jointklbalance}). The CAROLINA experiment involves the flexible analogue of stable balancing weights, using the $\chi^2$ divergence described in (\ref{sfbal}). $k$ is not manually set, but is chosen from the data according to the the residual bend analysis depicted in \Cref{fig:residualbends}. To handle $n < d$, we implement EBAL, SBW, and nCBPS ourselves, minimizing the penalized formulation (\ref{abstractpenaltybalancing}) using gradient descent.

\textbf{Optimization}. With fixed $\lambda$, the standard covariate balancing program (\ref{abstractpenaltybalancing}) is convex. Flexible covariate balancing (\ref{jointbalance}) is not convex, due to the product terms. We attempt to optimize (\ref{jointbalance}) using gradient descent. Instead of implementing the constraints on $W$, $\lambda$ and $\delta$ as projections, we embrace nonconvexity: we reparameterize $W = \mathrm{softmax}(\widetilde{W})$, $\lambda = \textrm{softplus}(\tilde{\lambda})$ and $\delta = \textrm{sigmoid}(\tilde{\delta})$ using new unconstrained variables $\widetilde{W} \in \mathbb{R}^{n_1}$ and $\tilde{\lambda},\tilde{\delta} \in \mathbb{R}$. We initialize $\widetilde{W} = 0$, $\tilde{\lambda} = 1$ and $\tilde{\delta} = 0$.

\textbf{Simulations}. Both the subgroup and celebrity simulations are run with $n = 1000, d = 3000$. All presented results are averaged over multiple independent runs. 



\end{document}